\documentclass[11pt]{article}

\usepackage{amssymb,amsmath,amsthm,latexsym}
\usepackage{xspace}

\usepackage{xcolor}
\usepackage{tikz}
 \usepackage{graphicx}
\usepackage{authblk}
 
\usepackage{subfig}

\usepackage{epic,eepic}

\usepackage{multirow}

\usetikzlibrary{positioning}
\usetikzlibrary{decorations.pathreplacing}
\usetikzlibrary{automata,positioning}
\usetikzlibrary{decorations.pathmorphing}
\usetikzlibrary{decorations.markings}
\usetikzlibrary{decorations}
\usetikzlibrary{arrows}
\usetikzlibrary{arrows.meta}
\usetikzlibrary{patterns}
\usetikzlibrary{calc}
\usetikzlibrary{shapes}
\usetikzlibrary{fadings,	shadings}

\theoremstyle{plain}

\newtheorem{theorem}{Theorem}[section]

\newtheorem{proposition}[theorem]{Proposition}

\newtheorem{example}[theorem]{Example}
\newtheorem{lemma}[theorem]{Lemma}
\newtheorem{corollary}[theorem]{Corollary}

\newcommand{\eqdef}{\stackrel{\text{def}}=}

\newcommand{\lang}{\mathrm{L}\xspace}

\newcommand{\strat}{g}

\newcommand{\tuple}[1]{\langle #1  \rangle}
\newcommand{\stam}[1]{}
\newcommand{\pair}{\tuple}
\newcommand{\zug}{\tuple}

\newcommand{\Nat}{\ensuremath{\mathbb{N}}\xspace}

\newcommand{\w}{\omega}

\newcommand{\fun}[3]{\ensuremath{#1\colon #2 \to #3}}

\newcommand{\A}{{\cal A}}
\newcommand{\B}{{\cal B}}
\newcommand{\C}{{\cal C}}
\newcommand{\D}{{\cal D}}

\newcommand{\As}{{\A_\strat}}
\newcommand{\Asm}{{\A_\strat^m}}

\renewcommand{\mod}[1]{\ (\mathrm{mod}\ #1)}

\newcommand{\todo}[1]{\textsf{\textcolor{blue}{To do: #1}}}

\newcommand{\DBP}{\textsf{DetByP}\xspace}

\renewcommand{\phi}{\varphi}

\tikzset{every picture/.append style={initial text={}}}

\tikzset{every loop/.append style={->,-{Latex[length=2.5mm]}}}

\tikzset{every initial by arrow/.style={->,-{Latex[length=2.5mm]}, initial distance=0.3}}

\tikzstyle{flow}  = [inner sep=0cm, node distance=0cm and 0cm]

\tikzstyle{toL} = [anchor=mid east, flow]
\tikzstyle{toR} = [anchor=mid west, flow]
\tikzstyle{toC} = [align=center, anchor=mid, flow]

\tikzstyle{rejecting}=[forbidden sign]

\tikzstyle{rejecting}=[accepting]

\tikzstyle{trans}=[draw,-{Latex[length=2.5mm]},auto]
\tikzstyle{trans2}=[draw,{Latex[length=2.5mm]}-{Latex[length=2.5mm]},auto]

\tikzstyle{letters}=[auto]
\tikzstyle{priority}=[scale=0.8, toC]

\tikzstyle{dot} = [draw,shape=circle,fill, minimum size=2mm, inner sep=0pt,outer sep=0pt]

\tikzstyle{memory}=[toC, scale=0.8, node distance=0.20]
\tikzstyle{looper}=[looseness=35, in=55, out=125]

\newcommand{\dpAut}{-0.35}

\newcommand{\state}[5]{

  \foreach \a/\vss/\mss/\hss in {#4/0.5/0.6/0.35} {
  \coordinate (center) at (#1,#2);

  \coordinate (RU\a) at ($(center) + (+\hss,+\vss*#5+\mss)$);
  \coordinate (RB\a) at ($(center) + (+\hss,-\vss*#5-\mss)$);

  \coordinate (LU\a) at ($(center) + (-\hss,+\vss*#5+\mss)$);
  \coordinate (LB\a) at ($(center) + (-\hss,-\vss*#5-\mss)$);

  \coordinate (opt\a) at ($(center) + (-0.0, -\vss*#5-\mss-0.15)$);

  \draw (RU\a) -- (RB\a) -- (LB\a) -- (LU\a) -- (RU\a);

  \node[below] at ($(center) + (0,\vss*#5+\mss)$) {#3};

  \foreach \x in {0, ..., #5} {
    \node[dot] (\a\x) at ($(center) + (0,\vss*#5) + (0,-\x * \vss * 2.0)$) {};
  }
  }
}

\theoremstyle{plain}
\newtheorem*{rep@theorem}{\rep@title}
\newcommand{\newreptheorem}[2]{%
\newenvironment{rep#1}[1]{%
\def\rep@title{#2 \ref{##1}}%
\begin{rep@theorem}}%
{\end{rep@theorem}}}

\newreptheorem{theorem}{Theorem}
\newreptheorem{lemma}{Lemma}
\newreptheorem{proposition}{Proposition}
\newreptheorem{conjecture}{Conjecture}
\newreptheorem{corollary}{Corollary}
\newreptheorem{claim}{Claim}
\newreptheorem{example}{Example}

\title{How Deterministic are Good-For-Games Automata?}

\author[1]{Udi Boker}
\author[2]{Orna Kupferman}
\author[3]{Micha{\l} Skrzypczak}

\affil[1]{Interdisciplinary Center, Herzliya, Israel\thanks{This research was supported by the Israel Science Foundation, grant no.\ 1373/16.}}
\affil[2]{The Hebrew University, Israel\thanks{This research has received funding from the European Research Council under the EU's 7-th Framework Programme (FP7/2007-2013) / ERC grant agreement no 278410.}}
\affil[3]{University of Warsaw, Poland\thanks{Supported by the Polish National Science Centre (decision UMO-2016/21/D/ST6/00491).}}

\date{}

\begin{document}

\maketitle

\vspace*{-8mm}
\begin{abstract}
In {\em good for games\/} (GFG) automata, it is possible to resolve nondeterminism in a way that only depends on the past and still accepts all the words in the language. The motivation for GFG automata comes from their adequacy for games and synthesis, wherein general nondeterminism is inappropriate.
We continue the ongoing effort of studying the power of nondeterminism in GFG automata. 
Initial indications have hinted that every GFG automaton embodies a deterministic one. Today we know that this is not the case, and in fact GFG automata may be exponentially more succinct than deterministic ones. 

We focus on the {\em typeness} question, namely the question of whether a GFG automaton with a~certain acceptance condition has an equivalent GFG automaton with a weaker acceptance condition on the same structure. 
Beyond the theoretical interest in studying typeness, its existence implies efficient translations among different acceptance conditions. This practical issue is of special interest in the context of games, where the B\"uchi and co-B\"uchi conditions admit memoryless strategies for both players. 
Typeness is known to hold for deterministic automata and not to hold for general nondeterministic automata.

We show that GFG automata enjoy the benefits of typeness, similarly to the case of deterministic automata. In particular, when Rabin or Streett GFG automata have equivalent B\"uchi or co-B\"uchi GFG automata, respectively, then such equivalent automata can be defined on a~substructure of the original automata. 
Using our typeness results, we further study the place of GFG automata in between deterministic and nondeterministic ones. Specifically, considering automata complementation, we show that GFG automata lean toward nondeterministic ones, admitting an exponential state blow-up in the complementation of a Streett automaton into a~Rabin automaton, as opposed to the constant blow-up in the deterministic case.
\end{abstract}

\section{Introduction}
Nondeterminism is a prime notion in theoretical computer science. It allows a computing  machine to examine, in a concurrent manner, all its possible runs on a certain input. For automata on finite words, nondeterminism does not increase the expressive power, yet it leads to an exponential succinctness~\cite{RS59}. For automata on infinite words, nondeterminism may increase the expressive power and also leads to an exponential succinctness. For example, nondeterministic B\"uchi automata are strictly more expressive than their deterministic counterpart~\cite{Lan69}. In the automata-theoretic approach to formal verification, we use automata on infinite words in order to model systems and their specifications. In particular, temporal logic formulas are translated to nondeterministic word automata~\cite{VW94}. In some applications, such as model checking, algorithms can proceed on the nondeterministic automaton, whereas in other applications, such as synthesis and control, they cannot. There, the advantages of nondeterminism are lost, and the algorithms involve a complicated determinization construction~\cite{Saf88} or acrobatics for circumventing determinization~\cite{KV05c}.
Essentially, the inherent difficulty of using nondeterminism in synthesis lies in the fact that each guess of the nondeterministic automaton should accommodate all possible futures. 


Some nondeterministic automata are, however, good for games:
in these automata it is possible to resolve the nondeterminism in a way that only depends on the past while still accepting all the words in the language. 
This notion, of  {\em good for games\/}  (GFG) automata was first introduced in~\cite{HP06}.\footnote{GFGness is also used in~\cite{CL10} in the framework of cost functions under the name ``history-determinism''.} Formally, a nondeterministic automaton $\A$ over an alphabet $\Sigma$ is GFG if there is a strategy $\strat$ that maps each finite word $u \in \Sigma^+$ to the transition to be taken after $u$ is read; and following $\strat$ results in accepting all the words in the language of $\A$. Note that a state $q$ of $\A$ may be reachable via different words, and $\strat$ may suggest different transitions from $q$ after different words are read. Still, $\strat$ depends only on the past, namely on the word read so far. Obviously, there exist GFG automata: deterministic ones, or nondeterministic ones that are {\em determinizable by pruning\/} (\DBP); that is, ones that just add transitions on top of a deterministic automaton. In fact, the GFG automata constructed in~\cite{HP06} are \DBP.\footnote{As explained in~\cite{HP06}, the fact that the GFG automata constructed there are \DBP does not contradict their usefulness in practice, as their transition relation is simpler than the one of the embodied deterministic automaton and it can be defined symbolically.}

Our work continues a series of works that have studied GFG automata: their expressive power, succinctness, and constructions for them, where the key challenge is to understand the power of nondeterminism in GFG automata.  
Let us first survey the results known so far. 
\stam{We first need to introduce some acceptance conditions for automata on infinite words. Since a run of an automaton on infinite words does not have a
final state, acceptance is determined with respect to the set of
states visited infinitely often during the run.  For example, in {\em B\"uchi\/}
automata, some of the states are designated as accepting states,
and a run is accepting iff it visits states from the accepting set
infinitely often~\cite{Buc62}.  Dually, in {\em co-B\"uchi\/}
automata, a run is accepting iff it visits states from the
accepting set only finitely often. More general are {\em Rabin\/}
automata. Here, the acceptance condition is a set
$\alpha=\{\zug{G_1,B_1},\ldots,\zug{G_k,B_k}\}$ of pairs of sets
of states, and a run is accepting if there is a pair
$\zug{G_i,B_i}$ for which the set of states visited infinitely
often intersects $G_i$ and is disjoint to $B_i$. The condition
$\alpha$ can also be viewed as a {\em Streett\/} condition, in
which case a run is accepting if for all pairs $\zug{G_i,B_i}$, if
the set of states visited infinitely often intersects $G_i$, then
it also intersects $B_i$. Finally, a {\em parity\/} condition maps each state to a rank in $\{1,\ldots,k\}$, and a run is accepting if the minimal rank that is visited infinitely often is even. 
The number $k$ of pairs or ranks in $\alpha$ is referred to as the {\em
index\/} of the automaton.
}%
In terms of expressive power, it is shown in~\cite{KSV06,NW98} that GFG automata with an acceptance condition of type $\gamma$ (e.g., B\"uchi) are as expressive as deterministic $\gamma$ automata.\footnote{The results in~\cite{KSV06,NW98} are given by means of {\em tree automata for derived languages\/}, yet, by~\cite{BKKS13}, the results hold also for GFG automata.}  Thus, as far as expressiveness is concerned, GFG automata behave like deterministic ones. The picture in terms of succinctness is diverse. For automata on finite words, GFG automata are always \DBP~\cite{KSV06,Mor03}. For automata on infinite words, in particular B\"uchi and co-B\"uchi automata\footnote{See Section~\ref{prelim aut} for the full definition of the various acceptance conditions.}, GFG automata need not be \DBP~\cite{BKKS13}. Moreover, the best known determinization construction of GFG B\"uchi automata is quadratic, whereas determinization of GFG co-B\"uchi automata has an exponential blow-up lower bound~\cite{KS15}.
Thus, in terms of succinctness, GFG automata on infinite words are more succinct (possibly even exponentially) than deterministic ones. 

For deterministic automata, where B\"uchi and co-B\"uchi automata
are less expressive than Rabin and Streett ones, researchers
have come up with the notion of an automaton being
{\em type}~\cite{KPB94}. 
Consider a deterministic automaton $\A$ with acceptance condition of type $\gamma$ and assume that $\A$ 
recognizes a language that can be recognized by some deterministic automaton with an acceptance condition of type $\beta$ that is weaker than $\gamma$. When deterministic $\gamma$~automata are $\beta$-type, it is guaranteed that a deterministic $\beta$-automaton for the language of $\A$ can be defined on top of the structure of $\A$.
For example, deterministic Rabin automata being B\"uchi-type \cite{KPB94} means that if a deterministic
Rabin automaton $\A$ recognizes a language that can be recognized
by a deterministic B\"uchi automaton, then $\A$ has an equivalent
deterministic B\"uchi automaton on the same structure. 
Thus, the basic motivation of typeness is to allow simplifications of the acceptance conditions of the considered automata without complicating their structure. Applications of this notion are much wider~\cite{KPB94}. 
In particular, in the context of games, the B\"uchi and co-B\"uchi conditions admit memoryless strategies for both players, which is not the case for the Rabin and Streett conditions~\cite{Tho95}. Thus, the study of typeness in the context of GFG automata addresses also the question of simplifying the memory requirements of the players. In addition, as we elaborate in Section~\ref{sec:consequences}, it leads to new and non-trivial bounds on the blow-up of transformations between GFG automata and their complementation.

Recall
that deterministic Rabin automata are B\"uchi-type.
Dually, deterministic Streett automata are co-B\"uchi-type.
Typeness 
can be defined also with respect to nondeterministic automata, yet it
crucially depends on the fact that the automaton is deterministic. Indeed, nondeterministic Rabin are not B\"uchi-type. Even with the co-B\"uchi acceptance condition, where  nondeterministic co-B\"uchi automata recognize only a subset of the $\omega$-regular languages, nondeterministic Streett automata are not co-B\"uchi-type~\cite{KMM06}. 

We first show that typeness is strongly related with determinism even when slightly relaxing the typeness notion to require the existence of an equivalent automaton on a substructure of the original automaton, instead of on the exact original 
structure, and even when we restrict attention to an {\em unambiguous\/} automaton, namely one that has a single accepting run on each word in its language. We describe an  unambiguous parity automaton $\A$, 
such that its language is recognized by a deterministic B\"uchi automaton, yet it is impossible to define a B\"uchi acceptance condition on top of a substructure of $\A$. We also point to a~dual result in~\cite{KMM06}, with respect to the co-B\"uchi condition, and observe that it applies also to the relaxed typeness notion.

We then show that for GFG automata, typeness, in its relaxed form, does hold. Notice that once considering GFG automata with no redundant transitions, which we call {\em tight\/}, the two typeness notions coincide. Obviously, all GFG automata can be tightened by removal of redundant transitions (Lemma~\ref{lem:Manipulations}).
In particular, we show that the typeness picture in GFG automata coincides with the one in deterministic automata: Rabin GFG automata are B\"uchi type, Streett GFG automata are co-B\"uchi type, and all GFG automata are type with respect to the weak acceptance condition. Unlike the deterministic case, however, the Rabin case is not a simple dualization of the Streett case; it is much harder to prove and it requires a stronger notion of tightness.

We continue with using our typeness results for further studying the place of GFG automata in between deterministic and nondeterministic ones.
We start with showing that all GFG automata that recognize languages that can be defined by deterministic weak automata are \DBP. This generalizes similar results about safe and co-safe languages~\cite{KMM06}. We then show that all unambiguous GFG automata are also \DBP.
Considering complementation, GFG automata lie in between the deterministic and nondeterministic settings---the complementation of a B\"uchi automaton into a co-B\"uchi automaton is polynomial, as is the case with deterministic automata, while the complementation of a co-B\"uchi automaton into a B\"uchi automaton as well as the complementation of a Streett automaton into a Rabin automaton is exponential, as opposed to the constant blow-up in the deterministic case. We conclude with proving a doubly-exponential lower bound for the translation of LTL into GFG automata, as is the case with deterministic automata.

The paper is structured as follows. In Section~\ref{sec:preliminaries} we provide the relevant notions about languages and GFG automata. Section~\ref{sec:examples} contains examples showing that typeness does not hold for the case of unambiguous automata. The next three sections, Sections~\ref{sec:NSWtoNCW}, \ref{sec:NRWtoNBW}, and~\ref{sec:ToWeak}, provide the main positive results of this work: co-B\"uchi typeness for GFG-Streett; B\"uchi typeness for GFG-Rabin; and weak typeness for GFG-B\"uchi and GFG-co-B\"uchi, respectively. Finally, in Section~\ref{sec:consequences} we continue to study the power of nondeterminism in GFG automata, looking into automata complementation and translations of LTL formulas to GFG automata.
 
\section{Preliminaries} 
\label{sec:preliminaries}


\subsection{Automata}
\label{prelim aut}
An automaton on infinite words is a tuple $\A=\tuple{\Sigma,Q,Q_0,\delta,\alpha}$, where $\Sigma$ is an input alphabet, $Q$ is a finite set of states,
$Q_0 \subseteq Q$ is a set of initial states, $\fun{\delta}{Q\times\Sigma}{2^Q}$ is a transition function that maps a state and a letter to a set of possible successors, and $\alpha$ is an acceptance condition. 
The first four elements, namely $\tuple{\Sigma,Q,\delta,Q_0}$, are the automaton's \emph{structure}.
We consider here the {\em B\"uchi\/}, {\em co-B\"uchi}, {\em parity}, {\em Rabin}, and {\em Streett} acceptance conditions. (The {\em weak\/} condition is defined in Section~\ref{sec:ToWeak}.)
In B\"uchi, and co-B\"uchi conditions, $\alpha \subseteq Q$ is a set of states. 
In a parity condition, $\fun{\alpha}{Q}{\{0,\ldots,k\}}$ is a function mapping each state to its priority.
In a Rabin and Streett conditions, $\alpha\subseteq 2^{2^{Q}\times 2^Q}$ is a set of pairs of sets of states. The {\em index\/} of a Rabin or Streett condition is the number of pairs in it. 
For a state $q$ of $\A$, we denote by $\A^q$ the automaton that is derived from $\A$ by changing the set of initial states to $\{q\}$. A {\em transition\/} of $\A$ is a triple $\zug{q,a,q'}$ such that $q' \in \delta(q,a)$. We extend $\delta$ to sets of states and to finite words in the expected way. Thus, for a set $S \subseteq Q$, a letter $a \in \Sigma$, and a finite word $u \in \Sigma^\ast$, we have that $\delta(S,\epsilon)=S$, $\delta(S,a)=\bigcup_{q \in S} \delta(q,a)$, and $\delta(S,u \cdot a)=\delta(\delta(S,u),a)$. Then, we denote by $\A(u)$ the set of states that $\A$ may reach when reading $u$. Thus, $\A(u)=\delta(Q_0,u)$.

Since the set of initial states need not be a singleton and the transition function may specify several successors 
for each state and letter, the automaton $\A$ may be {\em nondeterministic}. If $|Q_0|=1$ 
 and $|\delta(q,a)| \leq 1$ for every $q \in Q$ and $a\in \Sigma$, then $\A$ is {\em deterministic}.

Given an input word $w=a_1 \cdot  a_2 \cdots$ in
$\Sigma^{\omega}$, a {\em run} of $\A$ on $w$ is an infinite sequence $r=r_0,r_1,r_2,\ldots \in Q^\omega$ such that $r_0 \in Q_0$ and
for every $i \geq 0$, we have $r_{i+1} \in \delta(r_i,a_{i+1})$;
i.e., the run starts in the initial state and obeys the
transition function.  
For a run $r$, let $\inf(r)$ denote the set of states that $r$ visits
infinitely often. That is,
$\inf(r) = \{ q \in Q\mid\text{for infinitely many $i \geq 0$, we have $r_i=q$} \}$.

A set of states $S$ satisfies an acceptance condition $\alpha$ (or \emph{is accepting}) iff  
\begin{itemize}
\item $S \cap \alpha \neq \emptyset$, for a B\"uchi\ condition.
\item $S \cap \alpha = \emptyset$, for a co-B\"uchi\ condition.
\item $\min_{q\in\inf(r)}\{\alpha(q)\}$ is even, for a parity condition.
\item There exists $\tuple{E,F}\in\alpha$, such that $S \cap E = \emptyset$ and $S \cap F \neq \emptyset$ for a Rabin condition.
\item For all $\tuple{E,F}\in\alpha$, we have  $S \cap E = \emptyset$ or $S \cap F \neq \emptyset$ for a Streett condition.
\end{itemize}

Notice that B\"uchi\ and co-B\"uchi are dual, and so are Rabin and Streett. Also note that the B\"uchi and co-B\"uchi conditions are special cases of parity, which is a special case of Rabin and Streett. In the latter conditions, we refer to the sets $E$ and $F$ as the ``bad'' and ``good'' sets, respectively.
Finally, note that a Rabin pair may have an empty $E$ component, while an empty $F$ component makes the pair redundant (and dually for Streett).

A run $r$ is accepting if $\inf(r)$ satisfies $\alpha$. An automaton $\A$ accepts an input word $w$ iff there
exists an accepting run of $\A$ on $w$.
The {\em language} of $\A$, denoted by $\lang(\A)$, is the set of all words in
$\Sigma^\omega$ that $\A$ accepts. 
A nondeterministic automaton $\A$ is {\em unambiguous\/} if for every word $w \in \lang(\A)$, there is a single accepting run of $\A$ on $w$. Thus, while $\A$ is nondeterministic and may have many runs on each input word, it has only a single accepting run on words in its language. 

\stam{
A \emph{Rabin condition} $R$ and a \emph{Streett condition} $S$ of index $k$ over an automaton with states $Q$ are sets of $k$ pairs, $\pair{B_i,G_i}_{1 \leq i \leq k}$, where every \emph{bad set} $B_i$ and \emph{good set} $G_i$ are subsets of $Q$. A set $X$ of states is (Rabin) \emph{accepting} w.r.t.\ $R$ if there is a pair $\pair{B_i, G_i}$, such that $X\cap B_i=\emptyset$ and $X\cap G_i\neq\emptyset$. A set $X$ of states is (Streett) \emph{accepting} w.r.t.\ $S$ if for every pair $\pair{B_i, G_i}$, $X\cap B_i=\emptyset$ or $X\cap G_i\neq\emptyset$. Notice that a Rabin pair may have an empty $B_i$ component, while an empty $G_i$ component makes the pair redundant, while for Streett it is the other way around.
A set of states is \emph{rejecting} if it is not accepting. A path $\pi$ is accepting/rejecting if $\inf(\pi)$ is. 
}

We denote the different automata types by three-letter acronyms
in the set $\{ {\rm D, N} \} \times \{ {\rm B, C, P, R, S} \} \times \{
{\rm W} \}$. The first letter stands for the branching mode of
the automaton (deterministic or nondeterministic); the second for the acceptance-condition type (B\"uchi, co-B\"uchi, parity, Rabin, or Streett); and the third indicates that we consider automata on words. For Rabin
and Streett automata, we sometimes also indicate the index of the
automaton. In this way, for example, NBW are nondeterministic
B\"uchi word automata, and DRW[1] are deterministic Rabin automata
with index $1$.

For two automata $\A$ and $\A'$, we say that $\A$ and $\A'$ are
{\em equivalent\/} if $\lang(\A)=\lang(\A')$. For an automaton type
$\beta$ (e.g., DBW) and an automaton $\A$, we say that $\A$ is
$\beta$-realizable if there is a $\beta$-automaton equivalent to
$\A$.

Let $\A=\tuple{\A,Q,Q_0,\delta,\alpha}$ be an automaton.
For an acceptance-condition class $\gamma$ (e.g.,~B\"uchi), we say
that $\A$ is {\em $\gamma$-type} if $\A$ has an equivalent
$\gamma$ automaton with the same structure as $\A$~\cite{KPB94}. That is, there
is an automaton $\A'=\zug{\Sigma,Q,Q_0,\delta,\alpha'}$ such that
$\alpha'$ is an acceptance condition of the class $\gamma$ and
$\lang(\A')=\lang(\A)$. 

\subsection{Good-For-Games Automata}
\label{ssec:GFG}

An automaton $\A=\tuple{\Sigma,Q,Q_0,\delta,\alpha}$ is \emph{good for games} (GFG, for short) if there is a strategy $\fun{\strat}{\Sigma^\ast}{Q}$, such that for every word $w=a_1 \cdot a_2 \cdots \in \Sigma^\omega$, the sequence $\strat(w)=\strat(\epsilon)$, $\strat(a_1)$, $\strat(a_1 \cdot a_2)$,\ldots
is a run of $\A$ on $w$, and whenever $w\in \lang(\A)$, then $\strat(w)$ is accepting. We then say that $\strat$ \emph{witnesses} $\A$'s GFGness.

It is known~\cite{BKKS13} that if $\A$ is GFG, then its GFGness can be witnessed by a finite-state strategy, thus one in which for every state $q \in Q$, the set of words $\strat^{-1}(q)$ is regular. Finite-state strategies can be modeled by {\em transducers}. Given sets $I$ and $O$ of input and output letters, an {\em $(I/O)$-transducer\/} is a tuple ${\cal T}=\zug{I,O,M,m_0,\rho,\tau}$, where  $M$ is a finite set of states, to which we refer as \emph{memories}, $m_0\in M$ is an \emph{initial memory}, $\fun{\rho}{M\times I}{M}$ is a deterministic transition function, to which we refer as the \emph{memory update function}, and $\fun{\tau}{M}{O}$ is an output function that assigns a letter in $O$ to each memory. The transducer ${\cal T}$ generates a strategy $\fun{\strat_{\cal T}}{I^\ast}{O}$,  obtained by following $\rho$ and $\tau$ in the expected way: we first extend $\rho$ to words in $I^\ast$ by setting $\rho(\epsilon)=m_0$ and $\rho(u \cdot a)=\rho(\rho(u),a)$, and then define $\strat_{\cal T}(u)=\tau(\rho(u))$. 

Consider a GFG automaton $\A=\tuple{\Sigma,Q,Q_0,\delta,\alpha}$, and let $\strat=\langle\Sigma$, $Q$, $M$, $m_0$, $\rho$, $\tau\rangle$ be a finite-state $(\Sigma/Q)$-transducer that generates a strategy $\fun{\strat}{\Sigma^\ast}{Q}$ that witnesses $\A$'s GFGness (we abuse notations and use $\strat$ to denote both the transducer and the strategy it generates). 
Consider a state $q \in Q$. When $\tau(m)=q$, we say that $m$ is a {\em memory of $q$}. 
We denote by $\As$ the (deterministic) automaton that models the operation of $\A$ when it follows $\strat$. 
Thus, $\As=\zug{\Sigma,M,m_0,\rho,\alpha_\strat}$, where the acceptance condition $\alpha_\strat$ is obtained from $\alpha$ by replacing each set $F \subseteq Q$ that appears in $\alpha$ (e.g. accepting states, rejecting states, set in a Rabin or Streett pair, etc) by the set $F_\strat=\{m \mid \tau(m) \in F\}$. Thus, $F_\strat \subseteq M$ contains the memories of  $F$'s states.
For a state $q$ of $\A$, a path $\pi$ of $\As$ is \emph{$q$-exclusive accepting} if $\pi$ is accepting, and $\inf(\pi)\setminus\{m \ | \ m \mbox{ is a memory of }q\}$ is not accepting. 

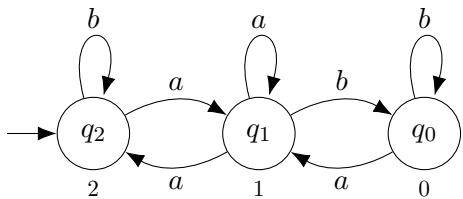
\begin{figure}
\centering
\begin{tikzpicture}[scale=2.2]

\node[state, initial] (q2) at (-1, 0) {$q_2$};
\node[state] (q1) at (0, 0) {$q_1$};
\node[state] (q0) at (+1, 0) {$q_0$};

\node at (-2.1, 0) {};
\node at (+2, 0) {};

\node[priority] at (-1, \dpAut) {$2$};
\node[priority] at (-0, \dpAut) {$1$};
\node[priority] at (+1, \dpAut) {$0$};

\path (q2) edge[trans, loop above] node[letters]{$b$} (q2);
\path (q1) edge[trans, loop above] node[letters]{$a$} (q1);
\path (q0) edge[trans, loop above] node[letters]{$b$} (q0);

\path (q2) edge[trans, bend left] node[letters, above]{$a$} (q1);
\path (q1) edge[trans, bend left] node[letters, below]{$a$} (q2);

\path (q1) edge[trans, bend left] node[letters, above]{$b$} (q0);
\path (q0) edge[trans, bend left] node[letters, below]{$a$} (q1);

\end{tikzpicture}
\caption{A weakly tight GFG-NPW $\A_0$. The numbers below the states describe their priorities.}
\label{fig:npw-gfg-example}
\end{figure}

\begin{example}
\label{ex1a}
Consider the NPW $\A_0$ appearing in Figure~\ref{fig:npw-gfg-example}. 
We claim that $\A_0$ is a GFG-NPW that recognizes the language
\[L_0=\{w \in\{a,b\}^\w\mid{} \text{there are infinitely many $b$'s in $w$}\}.\]
\end{example}

Indeed, if a word $w$ contains only finitely many $b$'s then $\A_0$ rejects $w$, as in all the runs of $\A_0$ on $w$, the lowest priority appearing infinitely often is~$1$. Therefore, $\lang(\A_0) \subseteq L_0$.

\begin{figure}
\begin{minipage}{6.0cm}
\begin{tikzpicture}[xscale=1.0, yscale=1.5]
\state{-2.2}{0}{$q_2$}{q2}{0}	
\state{+0.0}{0}{$q_1$}{q1}{1}
\state{+2.2}{0}{$q_0$}{q0}{0}

\draw ($(q20)+(-0.6,0)$) edge[trans] (q20);

\node[memory, below=of q00] {$m_0$};

\node[memory, below=of q10] {$m_1'$};
\node[memory, below=of q11] {$m_1$};

\node[memory, below=of q20] {$m_2$};

\node[priority] at (optq0) {$0$};
\node[priority] at (optq1) {$1$};
\node[priority] at (optq2) {$2$};

\path (q20) edge[trans, looper] node[letters] {$b$} (q20);
\path (q10) edge[trans, looper] node[letters] {$a$} (q10);
\path (q00) edge[trans, looper] node[letters] {$b$} (q00);

\path (q20) edge[trans] node[letters] {$a$} (q10);
\path (q10) edge[trans, bend left] node[letters] {$b$} (q00);
\path (q11) edge[trans, bend left] node[letters, below] {$b$} (q00);
\path (q00) edge[trans, bend left] node[letters] {$a$} (q11);
\path (q11) edge[trans] node[letters] {$a$} (q20);


\end{tikzpicture}
\caption{A strategy witnessing the GFGness of the automaton $\A_0$, depicted in Figure~\ref{fig:npw-gfg-example}.}
\label{fig:gfg-strat}
\end{minipage} \ \hspace{.3in}
\begin{minipage}{6.5cm}
\begin{tikzpicture}[xscale=1.0, yscale=1.5]
\state{-2.2}{0}{$q_2$}{q2}{0}	
\state{+0.0}{0}{$q_1$}{q1}{1}
\state{+2.2}{0}{$q_0$}{q0}{0}

\draw ($(q20)+(-0.6,0)$) edge[trans] (q20);

\node[memory, below=of q00] {$m_0$};

\node[memory, below=of q10] {$m_1'$};
\node[memory, below=of q11] {$m_1$};

\node[memory, below=of q20] {$m_2$};

\node[priority] at (optq0) {$0$};
\node[priority] at (optq1) {$1$};
\node[priority] at (optq2) {$2$};

\path (q20) edge[trans, looper] node[letters] {$b$} (q20);
\path (q10) edge[trans, looper] node[letters] {$a$} (q10);
\path (q00) edge[trans, looper] node[letters] {$b$} (q00);

\path (q20) edge[trans] node[letters] {$a$} (q10);
\path (q10) edge[trans, bend left] node[letters] {$b$} (q00);
\path (q00) edge[trans, bend left] node[letters] {$a$} (q10);


\end{tikzpicture}
\caption{A strategy witnessing the tightness of a sub-automaton of $\A_0$.}
\label{fig:gfg-strat2}
\end{minipage}
\end{figure}

We turn to describe a strategy $\fun{\strat}{\{a,b\}^\ast}{Q}$ with which $\A_0$ accepts all words in $L_0$. The only nondeterminism in $\A_0$ is when reading the letter $a$ in the state $q_1$. Thus, we have to describe $\strat$ only for words that reach $q_1$ and continue with an $a$. In that case, the strategy $\strat$ moves to the state~$q_2$, if the previous state is $q_0$, and to the state~$q_1$, otherwise. Figure~\ref{fig:gfg-strat} describes a $(\Sigma/Q)$-transducer that generates $\strat$. The rectangles denote the states of $\A_0$, while the dots are their $\strat$-memories. The numbers below the rectangles describe the priorities of the respective states of $\A_0$.

As $ L({\A_0}_\strat)\subseteq L(\A_0)$, it remains to formally prove that $L_0\subseteq L({\A_0}_\strat)$. Consider a word $w \in L_0$. Let $r=r_0,r_1,\ldots$ be the sequence of states of $\A_0$ visited by ${\A_0}_\strat$ on~$w$. Assume by way of contradiction that $r$ is not accepting. Thus, $r$ visits $q_1$ infinitely many times but visits $q_0$ only finitely many times. Let $N$ be such that $r_m \neq q_0$ for all $m\geq N$. Consider a position $k>N$ such that $r_k=q_1$. Since $w$ contains infinitely many $b$'s, there is some minimal $k'\geq k$ such that the $k'$-th letter in $w$ is $b$. Then, $r_k=r_{k+1}=\ldots=r_{k'}=q_1$ and $r_{k'+1}=q_0$, which contradicts the choice of $N$.\qed

The following lemma generalizes known residual properties of GFG automata (c.f.,~\cite{KS15}).

\newcommand{\lemMemoryTolerance}{
Consider a GFG automaton $\A=\tuple{\Sigma,Q,Q_0,\delta,\alpha}$ and let $\strat=\langle \Sigma,Q,M, m_0, \rho, \tau\rangle$ be a strategy witnessing its GFGness.
\begin{description}
\item[(1)]
For every state $q \in Q$ and memory $m \in M$ of $q$ that is reachable in $\As$, we have that $\lang(\Asm)=\lang(\A^q)$.
\item[(2)]
For every memories $m, m' \in M$ that are reachable in $\As$ with $\tau(m)=\tau(m')$, we have that $\lang(\Asm)=\lang(\A_\strat^{m'})$.
\end{description}
}

\begin{lemma}
\label{lem:MemoryTolerance}
\lemMemoryTolerance
\end{lemma} 

\begin{proof}
We start with the first claim. Obviously, $\lang(\Asm) \subseteq \lang(\A^q)$. For the other direction, consider toward contradiction that there is a word $w\in \lang(\A^q)\setminus \lang(\Asm)$. Let $u$ be a finite word such that $\As(u)=m$. Then,  $u \cdot w\not\in \lang(\As)$. However, there is an accepting run of $\A$ on $u \cdot w$: it  follows the run of $\As$ on $u$, and continues with the accepting run of $\A^q$ on $w$. Hence, $\strat$ does not witness $\A$'s GFGness, and we have reached a contradiction. The second claim is a direct corollary of the first, as $\lang(\Asm)=\lang(\A^{\tau(m)})=\lang(\A^{\tau(m')})=\lang(\A_\strat^{m'})$.
\end{proof}

A finite \emph{path} $\pi=q_0,\ldots,q_k$ in $\A$ is a sequence of states such that for $i=0,\ldots,k{-}1$ we have $q_{i+1}\in\delta(q_i,a_i)$ for some $a_i\in\Sigma$. A path is a \emph{cycle} if $q_0=q_k$. Each path $\pi$ induces a set ${\it states}(\pi)=\{q_0,\ldots,q_k\}$ of states in $Q$. A set $S$ of finite paths then induces the set ${\it states}(S)=\bigcup_{\pi \in S} {\it states}(\pi)$. For a set $P$ of finite paths, a {\em combination of paths from $P$} is a set ${\it states}(S)$ for some nonempty $S \subseteq P$. 

Consider a strategy $\strat=\langle \Sigma,Q,M, m_0, \rho, \tau\rangle$.
We say that a transition $\zug{q,a,q'}$ of $\A$ is {\em used by\/} $\strat$ if there is a word $u \in \Sigma^\ast$ and a letter $a \in \Sigma$ such that $q=\strat(u)$ and $q'=\strat(u \cdot a)$. 
Consider two memories $m \neq m' \in M$ with $\tau(m)=\tau(m')$. Let $P_{m'\to m}$ be the set of paths of $\As$ from $m'$ to $m$. We say that $m$ is {\em replaceable by\/} $m'$ if $P_{m' \to m}$ is empty or all combinations of paths from~$P_{m' \to m}$ are accepting. 

We say that $\A$ is {\em tight with respect to $\strat$\/} if all the transitions of $\A$ are used in $\strat$, and for all memories $m\neq m' \in M$ with $\tau(m)=\tau(m')$, we have that $m$ is not replaceable by $m'$. Intuitively, the latter condition implies that both $m$ and $m'$ are required in $\strat$, as an attempt to merge them strictly reduces the language of $\As$. 
When only the first condition holds, namely when all the transitions of $\A$ are used in $\strat$, we say that $\A$ is {\em weakly tight\/} with respect to $\strat$. 
When a Rabin automaton $\A$ is {\em tight with respect to $\strat$\/}, and in addition for every state $q$ that appears in some good set of $\A$'s acceptance condition, there is a $q$-exclusive accepting cycle in $\As$, we say that $\A$ is \emph{strongly tight} with respect to $\strat$.
Then, $\A$ is (weakly, strongly) {\em tight\/} if it is (weakly, strongly) tight with respect to some strategy.

\begin{example}
\label{ex1b}
The GFG-NPW $\A_0$ from Example~\ref{ex1a} is weakly tight and is not tight with respect to the strategy $\strat$. Indeed, while all the transitions in $\A_0$ are used in $\strat$, the memory $m_1$ is replaceable by $m_1'$, as all combinations of paths from $m_1'$ to $m_1$ are accepting. 
\hfill \qed
\end{example}

The following lemma formalizes the intuition that every GFG automaton can indeed be restricted to its tight part, by removing redundant transitions and memories. Further, every tight Rabin GFG automaton has an equivalent strongly tight automaton over the same structure.

\newcommand{\lemManipulations}{
For every GFG automaton $\A$ there exists an equivalent tight GFG automaton~$\A'$. Moreover, $\A'$ is defined on a substructure of $\A$. 
}

\begin{lemma}
\label{lem:Manipulations}
\lemManipulations
\end{lemma}

\begin{proof}
Consider a GFG automaton $\A=\tuple{\Sigma,Q,Q_0,\delta,\alpha}$, and let $\strat=\langle\Sigma$, $Q$, $M$, $m_0$, $\rho$, $\tau\rangle$ be a strategy that witnesses $\A$'s GFGness. We show how to make $\A$ tight with respect to a strategy obtained by merging memories in $\strat$. 

As long as $\A$ is not tight with respect to $\strat$, we proceed as follows. 
First, we remove from $\A$ all the transitions that are not used by $\strat$. Then, if there are two memories $m,m' \in M$ with $\tau(m)=\tau(m')$ such that $m$ is replaceable by $m'$, we remove $m$ from $\strat$ and redirect transitions to $m$ into $m'$. Note that the removal of $m$ may cause the obtained strategy not to use some transitions in $\A$. 
We thus keep repeating both steps as long as the obtained automaton is not tight with respect to the obtained strategy.

We prove that both steps do not change the language of $\A$ and its GFGness. 
First, clearly, removal of transitions that are not used does not change the language of $\A$. 
Now, consider memories $m\neq m' \in M$ with $\tau(m)=\tau(m')$ such that $m$ is replaceable by $m'$. Thus, $P_{m'\to m}$ is empty or all subsets $S \subseteq P_{m' \to m}$ are such that ${\it states}(S)$ is accepting. Let $\strat'$ be the strategy obtained by removing $m$ from $\strat$ and redirecting transitions to $m$ into $m'$.

Since $\lang(\A_{\strat'})\subseteq\lang(\A)=\lang(\A_\strat)$ it is enough to prove that $\lang(\A_\strat) \subseteq \lang(\A_{\strat'})$. 

We start with the case $P_{m'\to m}$ is empty, thus there is no path from $m'$ to $m$. Consider the accepting run $r$ of $\As$ on some word $w$. If $r$ does not include $m$, then the run of $\A_{\strat'}$ on $w$ is identical to $r$, and is thus accepting. Otherwise, let $p$ be the first position of $m$ in $r$, and let $w^{p+1}$ be the suffix of $w$ from the position $p+1$ onwards. Since $r$ is accepting, $w^{p+1}\in \lang(\Asm)$. Thus, by Lemma~\ref{lem:MemoryTolerance}, we have $w^{p+1}\in \lang(\A_\strat^{m'})$. Now, since $P_{m'\to m}$ is empty, the runs of $\A_\strat^{m'}$ and $\A_{\strat'}^{m'}$ are identical on $w^{p+1}$, and are thus accepting. Hence, $\A_{\strat'}$ accepts~$w$.

We continue with the case that all subsets $S \subseteq P_{m' \to m}$ are such that ${\it states}(S)$ is accepting. Consider a word $w\in\As$, and let $r'$ be the run of $\A_{\strat'}$ on $w$. The run $r'$ may use the memory $m'$ instead of $m$ finitely or infinitely many times. Consider first the case that $r'$ uses the memory $m'$ instead of $m$ for $k$ times. It is easy to prove, by an induction on $k$, that $r'$ is accepting. Indeed, the base case is similar to the case $P_{m'\to m}$ is empty, and the induction step changes only a finite prefix of the run. 
Consider now the case that the change is done infinitely many times, in positions $p_1, p_2, \ldots$ of $r'$. Every path from $p_i$ to $p_{i+1}$ is a path from $m'$ to $m$ in $\As$. Hence, the set of states $\inf(r')$ is ${\it states}(S)$ for some nonempty $S \subseteq P_{m' \to m}$, and is thus accepting.
\end{proof}

\newcommand{\lemStronglyTight}{
For every tight Rabin GFG automaton, there exists an equivalent strongly tight Rabin GFG automaton over the same structure. 
}

\begin{lemma}
\label{lem:StronglyTight}
\lemStronglyTight
\end{lemma}

\begin{proof}
Consider a tight GFG Rabin automaton $\A$ and let $\strat$ be a strategy that witnesses $\A$'s GFGness and with respect to which $\A$ is tight. We show that the removal of redundant states from the good sets of $A$'s accepting condition results in an automaton that is equivalent to $\A$ and strongly tight with respect to $\strat$.

Consider a state $q$ of $\A$ that appears in some good set $G$ of $\A$'s acceptance condition, and for which there is no $q$-exclusive accepting cycle in $\As$. We claim that the automaton $\A'$ that is identical to $\A$, except for removing $q$ from $G$, is a GFG Rabin automaton equivalent to $\A$ that is tight w.r.t.\ $\strat$. Indeed:
\begin{itemize}
\item Regarding the language equivalence, obviously, $L(\A') \subseteq L(\A)$. As for the other direction, let $r$ be the accepting run of $\As$ on some word $w$. Observe that $r$ is also an accepting run of $\A'_{\strat}$ on $w$: If $q$ does not appear infinitely often in $r$ then clearly $r$ is also accepting w.r.t.\ $\A'$. Now, if $q$ does appear infinitely often in $r$, then since there is no $q$-exclusive accepting cycle in $\As$, every cycle from $q$ back to $q$ is accepting w.r.t.\ $\A'$ and thus $r$ is accepting w.r.t.\ $\A'$.
\item Regarding the GFGness of $\A'$, since $L(\A) = L(\As) = L(\A'_\strat) \subseteq L(\A') \subseteq L(\A)$, we get that $\strat$ witnesses the GFGness of $\A'$. 
\item Regarding the tightness of $\A'$ w.r.t.\ $\strat$, observe that $\A$ and $\A'$ have the same transitions, and since $\As$ has no redundant memories, neither does $\A'_\strat$ have ones: Recall that a memory $m$ is redundant if exists a memory $m'$ of the same state, such that the set of paths of $\As$ from $m'$ to $m$, which we denote by $P_{m' \to m}$, is empty or all combinations of paths from~$P_{m' \to m}$ are accepting. The set of paths of $\As$ and of $\A'_\strat$ from $m'$ to $m$ are the same, and a path of $\A'_\strat$ cannot be accepting if it is not accepting in $\As$.
\end{itemize}

As there are finitely many states in $\A$, an iterative removal of states $q$ as described above results in an automaton that is strongly tight w.r.t.\ $\strat$.
\end{proof}

\begin{example}
\label{ex:making-tight}
In Figure~\ref{fig:gfg-strat2} we describe a strategy $\strat'$ that witnesses the tightness of a GFG-NPW on a substructure of the GFG-NPW $\A$ from Example~\ref{ex1a}. The strategy $\strat'$ is obtained from $\strat$ by following the procedure described in the proof of Lemma~\ref{lem:Manipulations}: all the transitions to $m_1$ are redirected to $m_1'$. This causes the transition $(q_1,a,q_2)$ that was used by the memory $m_1$ not to be used, and it is removed.
\hfill \qed
\end{example}

A special case of GFG automata are those who are \emph{determinizable by pruning} (or shortly \DBP) --- there exists a state $q_0\in Q_0$ and a function $\fun{\delta'}{Q\times\Sigma}{Q}$ that for every state $q$ and letter $a$ satisfies $\delta'(q,a)\in\delta(q,a)$ such that $\A'=\tuple{\Sigma,Q,q_0,\delta',\alpha}$ is a deterministic automaton recognizing the language $\lang(\A)$.

\section{Typeness Does Not Hold for Unambiguous Automata}
\label{sec:examples}

As noted in~\cite{KMM06}, it is easy to see that typeness does not hold for nondeterministic automata: 
there exists an NRW that recognizes an NBW-realizable language, yet does not have an equivalent NBW on the same structure. Indeed, since all $\omega$-regular languages are NBW-realizable, typeness in the nondeterministic setting would imply a translation of all NRWs to NBWs on the same structure, and we know that such a translation may involve a blow-up linear in the index of the NRW~\cite{SN99}. Even for Streett and co-B\"uchi automata, where the restriction to NCW-realizable languages amounts to a restriction to DCW-realizable languages, typeness does not hold. 

In this section we strengthen the relation between typeness and determinism and show that typeness does not hold for nondeterministic automata even when they recognize a~DBW-realizable language and, moreover, when they are unambiguous. Also, we prove the non-typeness results for NPWs, thus they apply to both Rabin and Streett automata. 

\begin{proposition}
\label{pro:unambig-not-type}
Unambiguous NPWs are not B\"uchi-type with respect to DBW-realizable languages.
\end{proposition}

\begin{figure}
\centering
\begin{tikzpicture}[scale=2.2]
\node[state, initial below] (q00) at (-1-1, 0) {$q_{00}$};
\node[state, initial right] (q01) at (-1-0, 0) {$q_{01}$};
\node[state] (q10) at (-1-1, -1) {$q_{10}$};
\node[state] (q11) at (-1-0, -1) {$q_{11}$};

\node[state, initial] (p0) at (+1+0, -0.5) {$p_0$};
\node[state] (p1) at (+1+1, -0.5) {$p_1$};
\node[state] (p2) at (+1+2, -0.5) {$p_2$};

\node[priority] at (-1-1+\dpAut*0.7, \dpAut*0.7) {$2$};
\node[priority] at (-1-0-\dpAut*0.7, \dpAut*0.7) {$1$};
\node[priority] at (-1-1+\dpAut*0.7, -1+\dpAut*0.7) {$1$};
\node[priority] at (-1-0-\dpAut*0.7, -1+\dpAut*0.7) {$1$};

\node[priority] at (+1+0, -0.5+\dpAut) {$1$};
\node[priority] at (+1+1, -0.5+\dpAut) {$1$};
\node[priority] at (+1+2, -0.5+\dpAut) {$0$};

\path (q00) edge[trans, loop left] node[letters]{$b$} (q00);
\path (q00) edge[trans, bend left=0] node[letters, above]{$b$} (q01);
\path (q00) edge[trans2] node[letters, below, pos=0.2]{$a$} (q11);
\path (q10) edge[trans2] node[letters, below, pos=0.8]{$a$} (q01);
\path (q10) edge[trans, bend right=0] node[letters, below]{$b$} (q11);
\path (q10) edge[trans, loop left] node[letters]{$b$} (q10);

\path (p0) edge[trans, loop above] node[letters]{$b$} (p0);
\path (p0) edge[trans, bend left] node[letters, above]{$a$} (p1);
\path (p1) edge[trans, loop above] node[letters]{$a$} (p1);
\path (p1) edge[trans, bend left] node[letters, above]{$b$} (p2);
\path (p2) edge[trans, bend left=40] node[letters, below]{$a$, $b$} (p0);

\end{tikzpicture}
\caption{$\A_1$: An unambiguous NPW that is DBW-realizable yet is not B\"uchi-type.}
\label{fig:upw-example}
\end{figure}
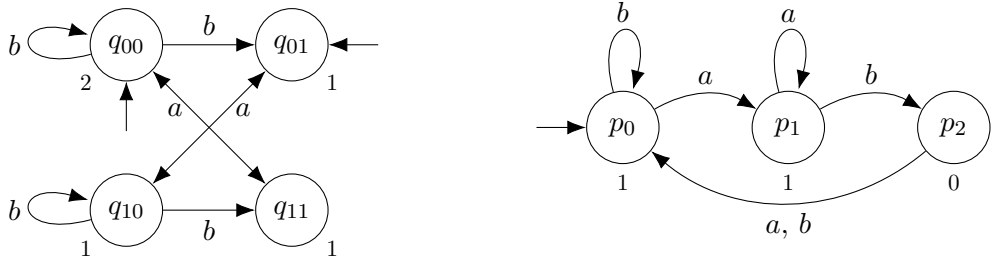

\begin{proof}
Consider the automaton $\A_1$ depicted in Figure~\ref{fig:upw-example}. We will show that $\A_1$ is unambiguous and recognizes a DBW-realizable language, yet $\A_1$ is not B\"uchi-type. Moreover, we cannot prune transitions from $\A_1$ and obtain an equivalent B\"uchi-type NPW.

The NPW $\A_1$ has two components: the left component, consisting of the states $q_{ij}$; and the right component, consisting of the states $p_0$, $p_1$, and $p_2$. The right part is deterministic, and it recognizes the language
\[L_{1,a,b}=\{w \in\{a,b\}^\w\mid\text{there are infinitely many $a$'s and $b$'s in $w$}\}.\]

We first prove that the left component is unambiguous and that its language is:
\[L_{1,\sharp a, b}=\{w \in\{a,b\}^\w\mid\text{there is a finite and even number of $a$'s in $w$ }\}.\]
To see this, observe that after reading a finite word, the left component of $\A_1$ can reach a state of the form $q_{ij}$ iff $i\equiv \sharp_a(w)\mod 2$ (i.e. $i$ is the parity of the number of letters $a$ in $w$). The only accepting runs of the left component are those that get stuck in the state $q_{00}$. This implies that if $w$ is accepted by the left component, then $w \in L_{1,\sharp a,b}$. For the other direction, consider a word $w \in L_{1,\sharp a,b}$. We show that $\A_1$ has an (in fact, unique) accepting run on $w$. 
We can construct an accepting run of the left component of $\A_1$ on $w$ by guessing whether the next block of $a$ (i.e., a sub-word of the form $a^{{+}}$) has an even or odd length. If the guess is incorrect, the run is stuck reading $b$ in a state of the form $q_{i1}$. If the guess is correct, the run reads the first $b$ after the block in a state of the form $q_{i0}$. Thus, after reading the last block of $a$'s, the constructed run reaches the state $q_{00}$, stays there forever, and $\A_1$ accepts $w$ in its left component. Further, all other runs that attempt to accept $w$ in the left component are doomed to get stuck. Thus, the left component is unambiguous. 

Since $L_{1,a,b}\cap L_{1,\sharp a,b}=\emptyset$, the NPW $\A_1$ is unambiguous and its language is
\begin{align*}
L_1=\{w \in\{a,b\}^\w\mid{}& \text{$w$ has an infinite number of $b$'s}\\
&\text{and an infinite or even number of $a$'s}\}.
\end{align*}

It is not hard to see that $L_1$ is DBW-realizable. An example of a DBW that recognizes $L_1$ is depicted in Figure~\ref{fig:dbw-example}.

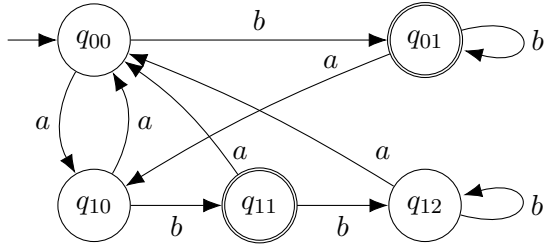
\begin{figure}
\centering
\begin{tikzpicture}[scale=2.2]
\node[state, initial] (q00) at (-2, 0) {$q_{00}$};
\node[state, accepting] (q01) at (0, 0) {$q_{01}$};
\node[state] (q10) at (-2, -1) {$q_{10}$};
\node[state, accepting] (q11) at (-1, -1) {$q_{11}$};
\node[state] (q12) at (0, -1) {$q_{12}$};

\draw[trans] (q00) edge[bend left=0] node[letters, above] {$b$} (q01);
\draw[trans] (q01) edge[out=+15, in=-15,loop] node[letters, right] {$b$} (q01);

\draw[trans] (q10) edge[bend right=0] node[letters, below] {$b$} (q11);
\draw[trans] (q11) edge[bend right=0] node[letters, below] {$b$} (q12);
\draw[trans] (q12) edge[out=-15, in=+15,loop] node[letters, right] {$b$} (q12);

\draw[trans] (q00) edge[bend right] node[letters, left] {$a$} (q10);
\draw[trans] (q10) edge[bend right] node[letters, right] {$a$} (q00);
\draw[trans] (q11) edge[bend right=10] node[letters, right, pos=0.15] {$a$} (q00);
\draw[trans] (q12) edge[bend right=5] node[letters, above right, pos=0.1] {$a$} (q00);
\draw[trans] (q01) edge[bend right=5] node[letters, above, pos=0.2] {$a$} (q10);
\end{tikzpicture}
\caption{A DBW recognizing $L_1$.}
\label{fig:dbw-example}
\end{figure}

We prove that $\A_1$ is not B\"uchi-type. Assume by way of contradiction that there exists a subset $\alpha$ of $\A_1$'s states such that the automaton obtained form $\A_1$ by viewing it as an NBW with the acceptance condition $\alpha$ recognizes $L_1$. If $\{q_{00},q_{11}\} \cap \alpha \neq \emptyset$, then the NBW accepts the word $a^\w$, which is not in $L_1$. If $\{q_{01},q_{10}\} \cap \alpha \neq \emptyset$, then the NBW accepts the word $ba^\w$, which is also not in $L_1$. Therefore, $\alpha \subseteq \{p_0,p_1,p_2\}$. Clearly, $p_1\notin \alpha$, as otherwise the NBW accepts $a^\w$. Similarly, if $p_0\in \alpha$, then the NBW accepts $ab^\w$, which is also not in $L_1$. Thus, $\alpha=\{p_2\}$ and the NBW rejects $b^\w$, which is in $L_1$.

Finally, as $\A_1$ is unambiguous and all its transitions are used in the accepting run of some word, it cannot be pruned to an equivalent NPW.
\end{proof}

The dual case of unambiguous NPWs that are not co-B\"uchi-type with respect to DCW-realizable languages follows from the results of~\cite{KMM06}, and we give it here for completeness, adding the observation that the automaton described there cannot be pruned to an equivalent co-B\"uchi-type NPW.

\begin{proposition}{\rm \cite{KMM06}}
\label{pro:unpw}
Unambiguous NPWs (and even NBWs) are not co-B\"uchi-type with respect to DCW-realizable languages.
\end{proposition}

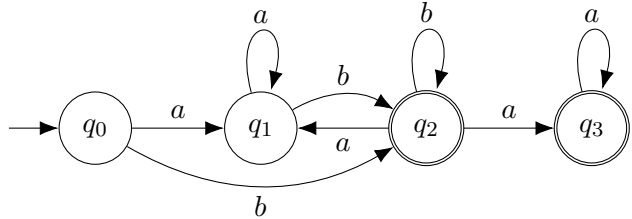
\begin{figure}
\centering
\begin{tikzpicture}[scale=2.2]
\node[state, initial] (q0) at (0, 0) {$q_0$};

\node[state] (q1) at (1, 0) {$q_1$};
\node[state, accepting] (q2) at (2, 0) {$q_2$};
\node[state, accepting] (q3) at (3, 0) {$q_3$};

\draw[trans] (q0) edge[bend right] node[letters, below]{$b$} (q2);
\draw[trans] (q0) edge[] node[letters, above]{$a$} (q1);
\draw[trans] (q2) edge[] node[letters]{$a$} (q3);

\draw[trans] (q2) edge[bend left=0] node[letters]{$a$} (q1);
\draw[trans] (q1) edge[bend left] node[letters]{$b$} (q2);

\draw[trans] (q2) edge[loop above] node[letters]{$b$} (q2);
\draw[trans] (q3) edge[loop above] node[letters]{$a$} (q3);
\draw[trans] (q1) edge[loop above] node[letters]{$a$} (q1);

\end{tikzpicture}
\caption{$\A_2$: An unambiguous NBW  that is DCW-realizable yet is not co-B\"uchi-type.}
\label{fig:NBWnonWT}
\end{figure}

\begin{proof}
Consider the NBW $\A_2$ depicted in Figure~\ref{fig:NBWnonWT}. We will show that $\A_2$ is unambiguous, and recognizes a DCW-realizable language, yet $\A_2$ is not co-B\"uchi-type. Moreover, we cannot prune transitions from $\A_2$ for obtaining an equivalent co-B\"uchi-type NPW.

Notice that $\lang(\A_2)= \{w\in\{a,b\}^\w\mid \text{$w$ contains a letter $b$}\}$, which is DCW-realizable.

Yet, there is no way to define a co-B\"uchi acceptance condition on top of $\A_2$ and obtain an equivalent NCW. Moreover, as $\A_2$ is unambiguous and all its transitions are used in an accepting run of some word, it cannot be pruned to an equivalent one.
\end{proof}

We conclude this section with the following rather simple proposition, showing that automata that are both unambiguous and GFG are essentially deterministic. Essentially, it follows from the fact that by restricting an unambiguous GFG automaton $\A$ to reachable and nonempty states, we obtain, by pruning, a deterministic automaton, which is clearly equivalent to $\A$. 

\newcommand{\proUnambigDBP}{
Unambiguous GFG automata are \DBP.
}

\begin{proposition}
\label{unambig dbp}
\proUnambigDBP
\end{proposition}

\begin{proof}
Let $\A$ be an unambiguous GFG automaton, witnessed by a strategy $\strat$ that starts in a state $q_0$. Without loss of generality, we can assume that $\lang(\A)\neq\emptyset$.
Let $\A'$ be the restriction of $\A$ to reachable and nonempty states (namely to reachable states $q$, such that $\lang(\A^q)\neq\emptyset$).
It is clear that $\A'$ is obtained from $\A$ by pruning and that $\lang(\A')=\lang(\A)$.

We prove that $\A'$ is deterministic. Note first that there is a single nonempty initial state. Indeed, assume toward contradiction that there is an initial state $q'_0\neq q_0$, from which $\A$ has a run accepting some word $w$. Since $\A$ has an accepting run on $w$ starting from $q_0$, as witnessed by $\strat$, we get a contradiction to its unambiguity.

Next, we prove that $\A'$ is deterministic by showing that for every finite word $u$ over which $\A$ can reach a nonempty state, we have $|\A'(u)|=1$. Let $q$ be the state that $\A_\strat$ reaches when reading $u$ and assume toward contradiction the existence of a state $q'\neq q$, such that $q'\in\A'(u)$. As $q'$ is nonempty, $\A^{q'}$ accepts some word $w$. However, since $uw\in\lang(\A)$, we have by the GFGness of $\A$ that $\A^q$ also accepts $w$. Hence, $\A$ has two different accepting runs on $uw$, contradicting its unambiguity.
\end{proof}

\section{Co-B\"uchi Typeness for GFG-NSWs}
\label{sec:NSWtoNCW}

\stam{
Recall that DSWs are co-B\"uchi-type: every DSW that recognizes a DCW-realizable language has an equivalent DCW on the same structure. On the other hand, by Example~\ref{ex:unpw}, NSWs (and even unambiguous ones) are not co-B\"uchi-type. Thus, co-B\"uchi typeness is strongly related with determinism.
}

In this section we study typeness for GFG-NSWs and show that, as is the case with deterministic automata, tight GFG-NSWs are co-B\"uchi-type. 
On a more technical level, the proof of Theorem~\ref{thm:NswNcw} only requires the GFG automata to be weakly tight (rather than fully tight), implying that  Theorem~\ref{thm:NswNcw}  can be strengthened in accordance. This fact is considered in Section~\ref{sec:NRWtoNBW}, where the typeness of GFG-NRWs is shown to require full tightness.

\newcommand{\thmNSWNCW}{
Tight GFG-NSWs are co-B\"uchi-type: 
Every tight GFG-NSW that recognizes a GFG-NCW-realizable language has an equivalent GFG-NCW on the same structure.
}

\begin{theorem}
\label{thm:NswNcw}
\thmNSWNCW
\end{theorem}

\begin{proof}
Consider a GFG-NSW $\A=\tuple{\Sigma,Q,Q_0,\delta,\alpha}$, with $\alpha=\{\zug{E_1,F_1}$, $\ldots$, $\zug{E_k,F_k}\}$. For $1 \leq i \leq k$, we refer to the sets $E_i$ and $F_i$ as the {\em bad\/} and {\em good\/} sets of $\alpha$, respectively.  
Let $\strat=\zug{\Sigma,Q,M,m_0,\rho,\tau}$ be a strategy that witnesses $\A$'s GFGness and such that $\A$ is tight with respect to $\strat$.
Formally, the automaton $\A'$ is defined as $\A$ with the co-B\"uchi acceptance condition
\[\alpha'\eqdef \{q\mid\text{all the cycles in $\A_\strat$ that go through a $\strat$-memory of $q$ are rejecting}\}.\]
We prove that $\lang(\A)=\lang(\A')$ and that $\A'$ is a GFG-NCW. 


Let $Q=\{q_1,\ldots,q_n\}$. We define a sequence of NSWs $\A_0,\A_1,\ldots,\A_{n}$ and prove that:
$\lang(\A)=\lang(\A_0)= \lang(\A_1)= \cdots = \lang(\A_{n})$; $\strat$ witnesses the GFGness of $\A_l$  for all $0 \leq l \leq n$; and $\A_{n}$ is essentially the NCW $\A'$. For all $0 \leq l \leq n$, the NSW $\A_l$ has the same structure as $\A$. The acceptance condition of $\A_l$ is $\alpha_l \cup \{\zug{\alpha'_l,\emptyset}\}$, where $\alpha_l$ and $\alpha'_l$ are defined as follows. 

First, $\alpha_0=\alpha$ and $\alpha'_0=\emptyset$. Thus, going form $\A$ to $\A_0$ we only add to $\alpha$ a redundant pair $\zug{\emptyset,\emptyset}$. 
Clearly, $\lang(\A)=\lang(\A_0)$ and $\A_0$ is GFG witnessed by $\strat$. 

For $1 \leq l \leq n$, we obtain $\alpha_l$ and $\alpha'_l$ from $\alpha_{l-1}$ and $\alpha'_{l-1}$ in the following way. First, we remove $q_l$ from all the bad sets in $\alpha_{l-1}$. Then, if $q_l \in \alpha'$, we add it to $\alpha'_l$. 

We now prove that $\lang(\A_l)=\lang(\A_{l-1})$ and that $\A_{l}$ is GFG witnessed by $\strat$. 

We distinguish between two cases. If $q_l \in \alpha'$, the proof is not hard:  adding $q_l$ to $\alpha'_l$ forces it to be visited only finitely often regardless of visits in the good sets. Thus, $\lang(\A_{l}) \subseteq \lang(\A_{l-1})$. 
In addition,  $\lang(\A_{l-1}) \subseteq \lang(\A_l)$, and $\strat$ witnesses also the GFGness of $\A_l$. Indeed, an accepting run in $\lang(\A_{l-1})$ remains accepting in $\lang(\A_l)$. To see this, assume by way of contradiction that there is a run $r$ that satisfies  
$\alpha_{l-1} \cup \{\zug{\alpha'_{l-1},\emptyset}\}$ yet does not satisfy $\alpha_l \cup \{\zug{\alpha'_l,\emptyset}\}$. 
Since $\alpha_l$ is easier to satisfy than $\alpha_{l-1}$, it must be that $r$ violates the pair $\zug{\alpha'_l,\emptyset}$. Since $r$ satisfies the pair $\zug{\alpha'_{l-1},\emptyset}$, it must visit $q_l$ infinitely often. Since, however, $q_l \in \alpha'$, the latter indicates that $r$ eventually traverses only rejecting cycles in $\A_\strat$ and is thus rejecting also in $\A_l$. 

 If $q_l \not \in \alpha'$, we proceed as follows.  Consider a state $q$ that has a memory with an accepting cycle, and let $\A'$ be the NSW that is derived from $\A$ by taking $q$ out of the bad sets. The change can obviously only enlarge the automaton's language.  Assume toward contradiction that there is a word $w\in \lang(\A')\setminus \lang(\A)$. Since $\lang(\A')\setminus \lang(\A)$ is an $\omega$-regular language, we may assume that $w$ is a lasso word, namely of the form $w=uv^\omega$.

As the only difference between $\A$ and $\A'$ is the removal of $q$ from bad sets, it follows that an accepting run $r$ of $\A'$ on $w$ visits $q$ infinitely often. Hence, there are positions $i$ and $j$ of $w$, such that: I) $r$ visits $q$ in both $i$ and $j$, II) the inner position within $v$ is the same in positions $i$ and $j$, and III) the cycle $C_r$ that $r$ goes through between positions $i$ and $j$ is accepting. 

Let $x$ be the prefix of $w$ up to position $i$ and $y$ the infix of $w$ between positions $i$ and $j$. Notice that $xy^\omega=uv^\omega=w$. Consider the run $r'$ of $\A'$ on $w$ that follows $r$ up to position $j$, and from there on forever repeats the cycle $C_r$. By the above definition of $i$ and $j$, the run $r'$ is also accepting.

Notice that since $w\not\in \lang(\A)$, it follows that $C_r$ is rejecting for $\A$. As the only difference between $\A$ and $\A'$ is the removal of $q$ from the bad sets,  it follows that combining $C_r$ with any cycle $C_a$ that contains $q$ and is accepting for $\A$, yields a cycle that is accepting for $\A$. Recall that $q$ has such an accepting cycle $C_a$, having that $C_r \cup C_a$ is accepting. 

Since NCW=DCW, there is a DCW $\D$ equivalent to $\A$. Let $n$ be the number of states in $\D$. Let $z$ be a finite word over which $\As$ makes the cycle $C_a$, and consider the word $e=x(y^nz^n)^\omega$. We claim that $e\in \lang(\A) \setminus \lang(\D)$, leading to a contradiction. 

As for the positive part, $e\in \lang(\A)$ by the run of $\A$ that reaches $q$ and then follows the $C_r$ and $C_a$ cycles.

Next, we show that $e \not\in \lang(\D)$. For every $i\in\Nat$, let $e_i=x (y^n z^n)^i y^n$ be a subword of $e$, and let $m_i=\As(e_i)$. 
Notice that $m_i$ belongs to some state $q_i$ of $\A$ and not necessarily to $q$. 
By~\cite{KS15}, the fact there exists a finite word $u$ such that $q,q_i\in\A(u)$, implies that  $\lang(\A^q)=\lang(\A^{q_i})$. Thus, since $q\in\A(e_i)$, we have, by Lemma~\ref{lem:MemoryTolerance}, that $\lang(\A_\strat^{m_i})=\lang(\A^{q_i})=\lang(\A^q)$.


Since $y^\omega\not\in \lang(\A^q)$ and $\lang(\A_\strat^{m_i})=\lang(\A^q)$, it follows that $\As$ does not accept $x (y^n z^n)^i y^\omega$. Hence, the run of $\D$ on $e$ must visit a rejecting state on every period between $e_i$ and $e_{i+1}$, implying that it is rejecting.

Finally, in $\alpha_n$ all the bad sets are empty. Also, $\alpha'_n=\alpha'$. Thus, $\A_n$ is really an NCW with acceptance condition $\alpha'$, i.e.~$\A'$.
\end{proof}

The following example shows that the weak tightness requirement cannot be omitted, even when the GFG-NSW is actually a GFG-NBW.

\newcommand{\exNonTightNbw}{
The automaton $\A_3$ depicted in Figure~\ref{fig:weakly-tight-example-nbw} is GFG-NBW and recognizes a~GFG-NCW-realizable language, yet $\A_3$ has no equivalent NCW on the same structure.
}

\begin{example}
\label{ex:non-tight-nbw}
\exNonTightNbw
\end{example}

\stam{
\begin{figure}
\centering
\begin{tikzpicture}[scale=2.2]
\begin{scope}
\clip (-2.3, 0.8) rectangle (0.7, -1.4);

\node[state, initial below] (q0) at (-1, 0) {$q_0$};

\node[state, accepting] (q1) at (-2, 0) {$q_1$};

\node[state] (p0) at (0, 0) {$p_0$};
\node[state, accepting] (p1) at (0, -0.9) {$p_1$};
\node[state, accepting] (p2) at (-1, -0.9) {$p_2$};

\path (q0) edge[trans, loop above] node[letters]{$b$} (q0);
\path (q0) edge[trans, bend left] node[letters]{$a$} (q1);
\path (q1) edge[trans, bend left] node[letters]{$a$} (q0);
\path (q0) edge[trans, bend left=0] node[letters]{$b$} (p0);

\path (p0) edge[trans, loop above] node[letters]{$b$} (p0);
\path (p0) edge[trans, bend left=0] node[letters]{$a$} (p1);
\path (p1) edge[trans, bend left] node[letters]{$a$} (p2);
\path (p2) edge[trans, bend left] node[letters]{$a$} (p1);
\path (p2) edge[trans, loop left] node[letters]{$b$} (p2);
\end{scope}
\end{tikzpicture}
\caption{$\A_3$: A GFG-NBW that is GFG-NCW-realizable yet is not co-B\"uchi-type.}
\label{fig:weakly-tight-example-nbw}
\end{figure}
}

First, it is not hard to see~that $\lang(\A_3) = (aa)^\omega + (aa)^\ast b^{{+}} aa (b + aa)^\omega\subseteq\{a,b\}^\omega$.

Notice that if we remove the transition $(q_0,b,q_0)$ then $\A_3$ becomes a deterministic automaton for the same language. In particular, $\A_3$ is GFG. 
Clearly the language of $\A_3$ is GFG-NCW-realizable---once the transition $(q_0,b,q_0)$ is removed, we can make $p_0$ the only rejecting state, and obtain an equivalent DCW.

Now assume toward contradiction that there exists an NCW $\A_3'$ equivalent to $\A_3$ over the whole structure of $\A_3$. Let $\alpha'$ be its acceptance condition. Observe that $q_0\notin\alpha'$ as otherwise $\A_3'$ rejects the word $a^\omega$. 
In that case $\A_3'$ accepts the word $b^\omega$, leading to a contradiction.\qed

\section{B\"uchi Typeness for GFG-NRWs}
\label{sec:NRWtoNBW}

Studying typeness for deterministic automata, one can use the dualities between the B\"uchi and co-B\"uchi, as well as the Rabin and Streett conditions, in order to relate the B\"uchi-typeness of DRWs with the co-B\"uchi typeness of DSWs. In the nondeterministic setting, we cannot apply duality considerations, as by dualizing a nondeterministic automaton, we obtain a universal one. As we shall see in this section, our inability to use dualization considerations is not only technical. There is an inherent difference between the co-B\"uchi typeness of GFG-NSWs studied in Section~\ref{sec:NSWtoNCW}, and the B\"uchi typeness of GFG-NRWs, which we study here. 
We first show that while the proof of Theorem~\ref{thm:NswNcw} only requires weak tightness, B\"uchi typeness requires full tightness.
 
The following example shows that tightness is necessary already for GFG-NCW that are GFG-NBW-realizable. 

\begin{example}
\label{ex:tightnessForGFG-NCW}
The automaton $\A_4$ depicted in Figure~\ref{fig:weakly-tight-example-ncw} is a weakly tight GFG-NCW that recognizes a GFG-NBW-realizable language, yet $\A_4$ has no equivalent GFG-NBW on the same structure.
\end{example}

\begin{figure}
\begin{minipage}{6.0cm}\begin{tikzpicture}[scale=2.2]

\node[state, initial, rejecting] (q0) at (3.5, 0) {$q_0$};

\node[state] (q1) at ($(q0)+(1, 0.0)$) {$q_1$};

\node[state, rejecting] (p0) at ($(q0) + (0, 1.0)$) {$p_0$};
\node[state] (p1) at ($(q0) + (1, 1.0)$) {$p_1$};

\path (q0) edge[trans, bend left=0] node[letters, above] {$a$, $b$} (q1);
\path (q1) edge[trans, loop right] node[letters] {$a$} (q1);
\path (q1) edge[trans, bend left] node[letters, below] {$b$} (q0);
\path (q0) edge[trans, bend left=0] node[letters] {$b$} (p0);
\path (q1) edge[trans, bend right=0] node[letters, above right] {$b$} (p0);

\path (p0) edge[trans, loop left] node[letters] {$b$} (p0);
\path (p0) edge[trans, bend left] node[letters, above] {$a$} (p1);
\path (p1) edge[trans, loop right] node[letters] {$a$, $b$} (p1);
\end{tikzpicture}
\caption{$\A_4$: A weakly tight 
GFG-NCW that is GFG-NBW-realizable yet is not B\"uchi-type.}
\label{fig:weakly-tight-example-ncw}
\end{minipage} \ \hspace{.3in}
\begin{minipage}{6.5cm}\begin{tikzpicture}[scale=2.2]
\begin{scope}
\clip (-2.3, 0.8) rectangle (0.7, -1.4);

\node[state, initial below] (q0) at (-1, 0) {$q_0$};

\node[state, accepting] (q1) at (-2, 0) {$q_1$};

\node[state] (p0) at (0, 0) {$p_0$};
\node[state, accepting] (p1) at (0, -0.9) {$p_1$};
\node[state, accepting] (p2) at (-1, -0.9) {$p_2$};

\path (q0) edge[trans, loop above] node[letters]{$b$} (q0);
\path (q0) edge[trans, bend left] node[letters]{$a$} (q1);
\path (q1) edge[trans, bend left] node[letters]{$a$} (q0);
\path (q0) edge[trans, bend left=0] node[letters]{$b$} (p0);

\path (p0) edge[trans, loop above] node[letters]{$b$} (p0);
\path (p0) edge[trans, bend left=0] node[letters]{$a$} (p1);
\path (p1) edge[trans, bend left] node[letters]{$a$} (p2);
\path (p2) edge[trans, bend left] node[letters]{$a$} (p1);
\path (p2) edge[trans, loop left] node[letters]{$b$} (p2);
\end{scope}
\end{tikzpicture}
\vspace{-0.7cm}
\caption{$\A_3$: A GFG-NBW that is GFG-NCW-realizable yet is not co-B\"uchi-type.}%
\label{fig:weakly-tight-example-nbw}%
\end{minipage}
\end{figure}

First notice that the language of $\A_4$ is $L_4= a^\w + a^\ast b^{{+}} a (a + b)^\w\subseteq\{a,b\}^\omega$. Moreover, if we remove the transitions $(q_0,b,q_1)$ and $(q_1,b,q_0)$, then $\A_4$ becomes a deterministic automaton for the same language. In particular, $\A_4$ is GFG. Clearly, $L_4$ is both DBW- and DCW-realizable.

Now assume toward contradiction that there exists an NBW $\A'_4$ equivalent to $\A_4$ over the (whole) structure of $\A_4$. Let $\alpha$ be its acceptance condition. Observe that the state $q_1$ must belong to $\alpha$, as otherwise $\A'_4$ rejects the word $a^\w$. But in that case, $\A'_4$ accepts the word $b^\w$, leading to a contradiction.\qed

We now proceed to our main positive result, obtaining the typeness of GFG-NRWs.

\newcommand{\thmGFGNRWNBW}{
Tight GFG-NRWs are B\"uchi-type: 
Every tight GFG-NRW that recognizes a GFG-NBW-realizable language has an equivalent GFG-NBW on the same structure.
}

\newcommand{\lemDBWequiv}{
Consider a GFG-NRW $\A$ that is GFG-NBW-realizable and a strategy $\strat$ that witnesses its GFGness.
\begin{enumerate}
\item A $\strat$-memory $m$ of a state $q$ of $\A$ cannot belong to both a $q$-exclusive accepting cycle and a rejecting cycle.
\item Consider $\strat$-memories $m$ and $m'$ of a state $q$ of $\A$, such that $m$ belongs to a $q$-exclusive accepting cycle and $m'$ belongs to a rejecting cycle. Let $P_{m\to m'}$ and $P_{m'\to m}$ be the sets of paths from $m$ to $m'$ and from $m'$ to $m$, respectively. Then $P_{m\to m'}$ or $P_{m'\to m}$ satisfies the following property: It is empty or every combination of its paths is accepting. Formally, for $P=P_{m\to m'}$ or $P=P_{m'\to m}$, we have that ${\it states}(P)=\emptyset$ or ${\it states}(S)$ is accepting for all $S \subseteq P$. 
\end{enumerate}
}

\newcommand{\lemManipulationConclusions}{
Consider a strongly tight GFG-NRW $\A$ that is GFG-NBW-realizable. Then, every state $q$ of $\A$ that appears in some good set has a single $\strat$-memory, and all the $q$-cycles in $\As$ are accepting, and at least one of them is $q$-exclusive.
}

\newcommand{\lemNoRejectingCycles}{
Consider a strongly tight GFG-NRW $\A$ that is GFG-NBW-realizable.  
Then, every state $q$ of $\A$ that appears in some good set does not belong to a rejecting cycle. 
}

\begin{theorem}
\label{thm:GFG-NRW-NBW}
\thmGFGNRWNBW
\end{theorem}


Consider a tight GFG-NRW $\A$ that recognizes a GFG-NBW-realizable language. Let $\strat$ be a strategy that witnesses $\A$'s GFGness and with respect to which $\A$ is tight. By Lemma~\ref{lem:StronglyTight}, we have a GFG Rabin automaton $\A'$ over the structure of $\A$ that is strongly tight with respect to $\strat$.
We define an NBW $\B$ on top of $\A$'s structure, setting its accepting states to be all the states that appear in ``good'' sets of $\A'$ (namely in the right components of the Rabin accepting pairs). 

Clearly, $\lang(\A') \subseteq \lang(\B)$, as $\B$'s condition only requires the ``good'' part of $\A'$'s condition, without requiring to visit finitely often in a corresponding ``bad'' set. We should thus prove that $\lang(\B) \subseteq \lang(\A')$ and that $\B$ is~GFG. Once proving the language equivalence, $\B$'s GFGness is straight forward, as the strategy $g$ witnesses it. The language equivalence, however, is not at all straightforward.
  
In order to prove that  $\lang(\B) \subseteq \lang(\A')$, we analyze the cycles of $\A'$ and of $\A'_\strat$, as expressed by the following lemmas.

\begin{lemma}
\label{lem:DBWequiv}
\lemDBWequiv
\end{lemma}

\begin{proof}
We start with the first claim. First, by~\cite{KSV06}, there is a DBW $\D$ equivalent to $\A$. 
Assume, by way of contradiction, that there are finite words $p$, $u$ and $v$, such that $\As(p)=m$, $\Asm(u)=m$ along a $q$-exclusive accepting cycle, and $\Asm(v)=m$ along a rejecting cycle. 

Let $n$ be the number of states in $\D$, and consider the word $w=p (u^n v^n)^\omega$. 
For every $i\geq 1$, the NBW $\D$ accepts $p (u^n v^n)^i u^\omega$. Hence, it is not hard to prove that $\D$ also accepts~$w$.  

On the other hand, we claim that $\A$ does not accept $w$. Indeed, since $v$ is a rejecting cycle that includes $q$, it must visit states in a bad set $B_i$ for every $i$ such that $q$ belongs to a good set $G_i$. As the cycle $u$ is $q$-exclusive accepting, we get that the cycle $u^nv^n$ is rejecting.

For the second claim, assume, by way of contradiction, that there are paths $\pi_1, \ldots, \pi_n \in P_{m\to m'}$ and paths $\pi'_1, \ldots, \pi'_{n'} \in P_{m'\to m}$, such that both sets of states: $\bigcup_{i=1}^n {\it states}(\pi_i)$ and $\bigcup_{i=1}^{n'} {\it states}(\pi'_i)$ are rejecting. Consider the path $\pi = \pi_1 \pi'_1 \pi_2 \pi'_2 \ldots \pi_n \pi'_{n'}$, where w.l.o.g.\ $\pi_n$ is repeated until reaching the larger index $n'$. Then, since the union of Rabin rejecting cycles is rejecting, $\pi$ is a rejecting cycle of $m$, contradicting the previous observation.
\end{proof}

\begin{lemma}
\label{lem:ManipulationConclusions}
\lemManipulationConclusions
\end{lemma}

\begin{proof}
Since $\A$ is strongly tight and $q$ appears in a good set, the ``strong tightening'' of $\A$, as per the proofs of Lemmas~\ref{lem:Manipulations} and \ref{lem:StronglyTight}, guarantees that $q$ has a $\strat$-memory $m$ that belongs to a $q$-exclusive accepting cycle. Assume, by way of contradiction, that $q$ has another memory $m' \neq m$. Then, due to the removal of redundant memories in Lemma~\ref{lem:Manipulations}, there is a rejecting combination of paths from $m$ to $m'$, as well as from $m'$ to $m$. Hence, $m$ belongs to a rejecting cycle, in contradiction to Lemma~\ref{lem:DBWequiv}.

In addition, since there is a single memory $m$ in $q$, and $q$ belongs to a good set, we have, by Lemma~\ref{lem:Manipulations}, that $m$ belongs to a $q$-exclusive accepting cycle. Hence, by Lemma~\ref{lem:DBWequiv}, the memory $m$ cannot also belong to a rejecting cycle.
\end{proof}


\begin{lemma}
\label{lem:NoRejectingCycles}
\lemNoRejectingCycles
\end{lemma}

\begin{proof}
Assume, by way of contradiction, that $q$ belongs to a rejecting cycle $\pi=q_0$, $q_1$, $q_2$,~\ldots, $q_n$, $q_{n+1}$ with $q_0=q_{n+1}=q$. Let $S$ be the set of indices of bad sets that $\pi$ visits. That is, an index $j$ belongs to $S$ if there is a state $p$ in $\pi$ that belongs to $B_j$. Notice that $S$ cannot be empty, since $q$ appears in a good set. 

Let $h$ be the maximal index of a state $q_i$ in $\pi$ up to which the strategy may exhaust the cycle states, while not adding a ``fresh unrejected good state''. That is:
\begin{itemize}
\item There is a path $\rho$ of $\As$ from $q$ to a memory $m$ of $q_h$ that visits $q_i$ for every $1 \leq i \leq h$, and if a state $p$ appears in $\rho$ and in $G_i\setminus B_i$ for some acceptance set $i$, then $i\in S$. (Notice that the path may also visit states not in the cycle and may visit the cycle states in a different order.)
\item There is no such path of $\As$ from $q$ to $q_{h+1}$. 
\end{itemize}
Notice that $h\geq 1$, since there is a transition $q \to q_1$ that the strategy uses, and $h \leq n$, since otherwise $q$ belongs to a rejecting path of $\As$, while such a path does not exist due to Lemma~\ref{lem:ManipulationConclusions}.

Let $m'$ be a memory of $q_h$ that takes the transition $q_h \to q_{h+1}$. Notice that $m'\neq m$, since by the maximality of $h$, $m$ does not take the transition $q_h \to q_{h+1}$. 

Furthermore, there cannot be rejecting path combinations from both $m$ to $m'$ and from $m'$ to $m$, as merging them would provide a rejecting path $\rho'$ from $m$ to $m'$, which is impossible due to the maximality of $h$. (Concatenating $\rho'$ to $\rho$ provides a continuation of $\rho$ to $q_{h+1}$.)

Hence, all path combinations from either $m$ to $m'$ or from $m'$ to $m$ are accepting. However, this leads to a contradiction due to the removal of redundant memories in Lemma~\ref{lem:Manipulations}. 
\end{proof}

We are now in position to finish the proof of Theorem~\ref{thm:GFG-NRW-NBW} by showing that $\lang(\B) \subseteq \lang(\A')$ and that $\B$ is~GFG.

Consider a word $w\in \lang(\B)$, and an accepting run $r$ of $\B$ on it. Let $q$ be an accepting state that appears infinitely often in $r$. 
By Lemma~\ref{lem:NoRejectingCycles}, all cycles of $\A'$ that include $q$ are accepting. Hence, $r$ is also an accepting run of $\A'$ on $w$.

As for the GFGness of $\B$, we claim that the strategy $\strat$ also witnesses $\B$'s GFGness. Consider a word $w\in \lang(\B)$. Since $\lang(\B) = \lang(\A') = \lang(\A_\strat')$, there is an accepting run $r$ of $\A_\strat'$ on $w$. Therefore, there must be some state $q$ in a good set of $\A'$ that is visited infinitely often along $r$. Thus, $r$ is also an accepting run of $\B_\strat$ on $w$. This concludes the proof of Theorem~\ref{thm:GFG-NRW-NBW}.\qed


The following result follows directly from Lemma~\ref{lem:Manipulations}, Theorem~\ref{thm:GFG-NRW-NBW}, and the determinization procedure for B\"uchi GFG automata from~\cite{KS15}.

\begin{corollary}\label{col:GFG-NRW-to-DBW}
Every GFG-NRW with $n$ states that recognizes a DBW-realizable language has an equivalent DBW with at most $O(n^2)$ states.
\end{corollary}

\begin{proof}
Consider a GFG-NRW $\A$ with $n$ states that recognizes a DBW-realizable language. By Lemma~\ref{lem:Manipulations}, $\A$ has an equivalent tight GFG-NRW on a substructure of it, thus with at most $n$ states. 
By Theorem~\ref{thm:GFG-NRW-NBW}, $\A$ has an equivalent GFG-NBW on the same structure, thus with at most $n$ states. By~\cite{KS15}, GFG-NBWs can be determinized with a quadratic blow-up, and we are done.
\end{proof}

\section{Weak Typeness for GFG Automata}
\label{sec:ToWeak}

A B\"uchi automaton $\A$ is {\em weak}~\cite{MSS88} if for each strongly connected component $C$ of $\A$, either $C \subseteq \alpha$ (in which case we say that $C$ is an {\em accepting component}) or $C\cap \alpha = \emptyset$ (in which case we say that $C$ is a {\em rejecting component}). Note that a weak automaton can be viewed as both a B\"uchi and a co-B\"uchi automaton, as a run of $\A$ visits $\alpha$ infinitely often iff it gets trapped in an accepting component iff it visits states in $Q \setminus \alpha$ only finitely often.
We use NWW and DWW to denote nondeterministic and deterministic weak word automata, respectively. 

We show in this section that all GFG automata are type with respect to the weak acceptance condition. We provide the theorem with respect to GFG-NCWs, from which we can easily deduce it, by our previous typeness results, also for the other types.

\newcommand{\thmGFGNCWDBW}{
Tight GFG-NCWs are weak-type:
every tight GFG-NCW that recognizes a~GFG-NWW-realizable language has an equivalent GFG-NWW on the same structure.
}

\begin{theorem}
\label{thm:GFG-NCW-NBW}
\thmGFGNCWDBW
\end{theorem}

\begin{proof}
Consider a tight GFG-NCW $\A$ that recognizes a language that is GFG-NWW-realizable.
Let $S$ be the set of rejecting states of $\A$ and let $\strat$ be a strategy witnessing $\A$'s tight GFGness.
Let $S'$ be the union of $S$ and all the states $q$ of $\A$ for which no $\strat$-memory $m$ has an accepting cycle in $\A_{\strat}$. Let $\A'$ be the automaton $\A$ with the co-B\"uchi condition given by $S'$. Notice that the strategy $\strat$ witnesses that for every state $q$ of $\A'$ we have $\lang(\A^q)\subseteq\lang\big((\A')^q\big)$. The opposite inclusion follows from the fact that $S\subseteq S'$. Thus, $\A'$ is an NCW equivalent to $\A$ and $\strat$ witnesses its GFGness. 

We now prove that $\A'$ is weak. Assume contrarily that there exists a cycle $C$ in $\A'$ that contains both a state $q\notin S'$ and a state $q'\in S'$.

Since $q\notin S'$, there is a cycle $C_{{+}}$ in $\A_{\strat}$ that is accepting in $\A_{\strat}$ and contains a $\strat$-memory $m$ of $q$. This cycle witnesses that none of the states on $C_{{+}}$ can belong to $S'\setminus S$, therefore the cycle $C_{{+}}$ is accepting in $\A'_{\strat}$ as well.

We construct a cycle $C_{{-}}$  of $\A'_{\strat}$ that visits some $\strat$-memory $m'$ of $q'$ and the $\strat$-memory $m$ of $q$. This cycle is obtained by extending the cycle $C$ of $\A'$ in the following way. Assume that $(q_0,a_0,q_1)$ and $(q_1,a_1,q_2)$ are two consecutive transitions of $C$. Since $\A'$ contains only transitions of $\strat$, these are actually transitions of $\A'_{\strat}$ of the form $(m_0,a_0,m_0')$ and $(m_1,a_1,m_1')$ with $\strat$-memories: $m_0$ of $q_0$; $m_0'$ and $m_1$ of $q_1$; and $m_1'$ of $q_2$. Notice that $m_0'$ may possibly be different from $m_1$. However, by the assumption that $\A$ is tight, there is a path in $\A'_{\strat}$ leading from $m_0'$ to $m_1$. Thus, for each pair of such consecutive transitions we can add an appropriate path to $C$ in such a way to obtain a cycle $C_{{-}}$ of $\A'_{\strat}$ that extends (as a set of states) $C$. Additionally, we can add to $C_{{-}}$ two paths in such a way to visit $q$ exactly in the $\strat$-memory $m$ ($C$ visits $q$, so it is possible as above). As $q'\in S'$ and $q'\in C\subseteq C_{{-}}$, we know that $C_{{-}}$ is rejecting in $\A'_{\strat}$.

Let $u_{{+}}$ and $u_{{-}}$ be the finite words over which ${(\A')}_{\strat}^{m}$ traverses the cycles $C_{{+}}$ and $C_{{-}}$, respectively. An infinite repetition of $u_{{+}}$ and $u_{{-}}$ belongs to $\lang\big((\A')_{\strat}^{m}\big)=\lang\big((\A')^q\big)=\lang\big(\A^q\big)$ if and only if it contains only finitely many copies of $u_{{-}}$. But this contradicts the fact that $\lang(\A^q)$ can be recognized by a DWW.
\end{proof}

Consider now a GFG-NSW $\A$ that is GFG-NWW-realizable. Notice that it is obviously also GFG-NBW-realizable. Hence, by Theorem~\ref{thm:NswNcw}, there is a GFG-NCW on $\A$'s structure, and by Theorem~\ref{thm:GFG-NCW-NBW} also a GFG-NWW. The cases of a GFG-NPW and a GFG-NBW obviously follow, since they are special cases of GFG-NSWs. As for a GFG-NRW $\A$ that is GFG-NWW-realizable, notice that it is obviously also GFG-NBW-realizable. Hence, by Theorem~\ref{thm:GFG-NRW-NBW}, there is a GFG-NBW on $\A$'s structure, and by Theorem~\ref{thm:GFG-NCW-NBW} also a GFG-NWW.

\begin{corollary}
Tight GFG-NSWs and GFG-NRWs are weak-type:
every tight GFG-NSW and GFG-NRW that recognizes a GFG-NWW-realizable language has an equivalent GFG-NWW on the same structure.
\end{corollary}

Next, we show that GFG-NWWs are \DBP, generalizing a folklore result about safe and co-safe GFG automata.

\newcommand{\thmGFGNWW}{
GFG-NWWs are \DBP.
}

\begin{theorem}
\label{thm:GfgNww}
\thmGFGNWW
\end{theorem}

\begin{proof}
Consider a GFG-NWW $\A$ with accepting set $\alpha$. By Lemmas~\ref{lem:Manipulations} and \ref{lem:StronglyTight}, we may assume that $\A$ is strongly tight w.r.t.\ a strategy $\strat$. 
First notice that by Lemma~\ref{lem:ManipulationConclusions}, a state $q\in\alpha$ has only one $\strat$-memory, and is therefore already deterministic.
  
Now we consider the case of a state $q\notin \alpha$ such that there are at least two $\strat$-memories $m$ and $m'$ of $q$. 
Let $\strat'$ be the strategy obtained by removing $m'$ from $\strat$ and redirecting transitions to $m'$ into $m$.

We now show that $\lang(\A_\strat) = \lang(\A_{\strat'})$. From that, by induction it follows that the number of memories of each state of $\A$ can be reduced to $1$.

Consider a word $w\in\As$, and let $r'$ be the run of $\A_{\strat'}$ on $w$. The run $r'$ may use the memory $m$ instead of $m'$ finitely or infinitely many times. 
If $r'$ uses it only finitely many times, then by an argument similar to the one given in the proof of Lemma~\ref{lem:Manipulations}, $r'$ is also accepting. (The argument inductively uses Lemma~\ref{lem:MemoryTolerance}, according to which $\lang(\A_\strat^{m})=\lang(\A_\strat^{m'})$.)

We continue with the case that the change is done infinitely many times, in positions $p_1, p_2, \ldots$ of $r'$, and assume toward contradiction that $r'$ is rejecting. Every path from $p_i$ to $p_{i+1}$ is a path from $m$ to $m'$ in $\As$. Notice that the suffix of $w$ from position $p_1$ onwards is in $\lang(\A^q)\setminus\lang(\A_{\strat'}^{m})$.
Since we consider $\w$-regular languages, we can assume without loss of generality that this suffix is periodic, in the form of $u^\w$, where $\Asm(u)=m'$. Let $u'$ be a finite word such that $\A_{\strat}^{m'}(u')=m$. 

Consider now a word $w'\in(u+u')^\w$.
First assume that $w'$ contains only finitely many instances of $u'$. In that case, we have that $\A^q$ accepts $w'$, because $\A^q$ has a run that loops back to $q$ until reading the last occurrence of $u'$ and then follows the run witnessing that $u^\w\in\lang(\A^q)$. We refer to such words as of the \emph{first kind}.

Now assume that a suffix of $w'$ from some point on is equal to $(u u')^\w$. Since $(u u')^\w\not\in\lang(\Asm)$, we get by Lemma~\ref{lem:MemoryTolerance} that $w'\notin\lang(\A^q)$. We refer to such words as of the \emph{second kind}.

Now consider the minimal, namely last, strongly-connected component of $\As$ that can be reached from $m$ by reading words in the language $(u+u')^\ast$. If this component is accepting, then $\Asm$ accepts a word of the second kind. Similarly, if the component is rejecting then $\Asm$ rejects a word of the first kind. In both cases we get a contradiction.
\end{proof}

By combining the above results, we obtain the following corollary.

\begin{corollary}
\label{cor nww}
Every GFG-NSW and GFG-NRW that recognizes a GFG-NWW-realizable language is \DBP.
\end{corollary}

\section{Consequences}
\label{sec:consequences}
GFG automata provide an interesting formalism in between deterministic and nondeterministic automata. Their translation to deterministic automata is immediate for the weak condition (Theorem~\ref{thm:GfgNww}), polynomial for the B\"uchi condition~\cite{KS15}, and exponential for the co-B\"uchi, parity, Rabin, and Streett conditions~\cite{KS15}. They have the same typeness behavior as deterministic automata, summarized in Table~\ref{table:Typeness}. The positive results of Table~\ref{table:Typeness} follow from our theorems in Sections~\ref{sec:NSWtoNCW}, \ref{sec:NRWtoNBW}, and~\ref{sec:ToWeak}. The negative results follow from corresponding counterexamples with deterministic automata~\cite{KPB94,KMM06}. Considering the complementation of GFG automata, they lie in between the deterministic and nondeterministic settings, as shown in Table~\ref{table:Complementation}. As for the translation of LTL formulas to GFG automata, it is doubly exponential, like the translation to deterministic automata (Corollary~\ref{cor:ltl-to-gfgnsw} below). 

\paragraph*{Complementation}
In the deterministic setting, Rabin and Streett automata are dual: complementing a DRW into a DSW, and vice versa, is simply done by switching between the two acceptance conditions on top of the same structure. This is not the case with GFG automata. We show below that complementing a GFG-NSW, and even a GFG-NCW, into a GFG-NRW involves an exponential state blow-up. Essentially, it follows from the B\"uchi-typeness of GFG-NRWs (Theorem~\ref{thm:GFG-NRW-NBW}) and the fact that while determinization of GFG-NBWs involves only a quadratic blow-up, determinization of GFG-NCWs involves an exponential one \cite{KS15}.

\newcommand{\corComplementation}{
The complementation of a GFG-NCW into a GFG-NRW involves a $2^{\Omega(n)}$ state blow-up.
}

\begin{corollary}\label{cor:Complementation}
\corComplementation
\end{corollary}

\begin{proof}
By~\cite{KS15}, there is a GFG-NCW $\C$ with $n$ states whose equivalent DCWs must have at least $2^{\Omega(n)}$ states.
Consider a GFG-NRW $\A$ with $x$ states for the complement of $\C$.

Since the language of $\A$ is DBW-recognizable, then, by Corollary~\ref{col:GFG-NRW-to-DBW}, there is a DBW $\D$ equivalent to $\A$ whose state space is quadratic in the number of states of $\A$, namely with up to $x^2$ states.
As the dual of $\D$ is a DCW equivalent to $\C$, it follows that $\D$ has at least $2^{\Omega(n)}$ states.
Hence, $x^2 \geq 2^{\Omega(n)}$, implying that $x \geq 2^{\Omega(n/2)} = 2^{\Omega(n)} $.
\end{proof}

Using our typeness results, we get an almost complete picture on complementation of GFG automata.

\newcommand{\DR}[1]{\multirow{2}{*}{#1}}
\newcommand{\BAR}[1]{\DR{$\overline{\mbox{#1}}$}}

\begin{table}
\parbox{.48\linewidth}{
\centering
\bgroup
\def\arraystretch{1.25} 
\begin{tabular}{c||c|c|c|c|c|c|}
\multicolumn{1}{r||}{\DR{Type To}} & \DR{\!W\!} & \DR{B} & \DR{C} & \DR{P} & \DR{R} & \DR{S} \\ 
\multicolumn{1}{c||}{From} & &&&&&   \\ 
\hline\hline
Weak & \multicolumn{6}{c|}{}\\
B\"uchi & \multicolumn{6}{c|}{} \\ 
Co-B\"uchi & \multicolumn{6}{c|}{Yes}\\ 
Parity & \multicolumn{6}{c|}{}\\ \cline{4-7}
Rabin & \multicolumn{2}{c|}{}&\multicolumn{2}{c|}{No}&Y&N\\  \cline{3-7}
Streett & & N &Y&\multicolumn{2}{c|}{No}&Y\\
 \hline
\end{tabular}
\caption{Typeness in translations between GFG automata. (Y=Yes; N=No.)}
\label{table:Typeness}
\egroup
}
\hfill
\parbox{.48\linewidth}{
\bgroup
\def\arraystretch{1.25} 
\begin{tabular}{c||c|c|c|c|c|c|}
\multicolumn{1}{r||}{\DR{Comp.\ To}} & \DR{\!$\overline{\mbox{W}}$\!} & \BAR{C} & \BAR{B} & \BAR{P} & \BAR{R} & \BAR{S} \\ 
\multicolumn{1}{c||}{From} & &&&&&   \\ 
\hline\hline
Weak & \multicolumn{6}{c|}{}\\
B\"uchi & \multicolumn{2}{c}{Poly}&\multicolumn{4}{c|}{} \\ \cline{4-7}
Co-B\"uchi & \multicolumn{2}{c|}{}&\multicolumn{3}{c|}{}&\multirow{3}{*}{?}\\ 
Parity & \multicolumn{2}{c|}{}&\multicolumn{3}{c|}{}&\\
Rabin & \multicolumn{2}{c|}{}&\multicolumn{3}{c|}{}&\\  \cline{3-3} \cline{7-7}
Streett & &\multicolumn{5}{c|}{\multirow{-4}{*}{Exp}}\\
 \hline
\end{tabular}
\caption{The state blow-up involved in the complementation of GFG automata.}
\label{table:Complementation}
\egroup
}
\end{table}

\begin{theorem}
The state blow-up involved in the complementation of GFG automata is as summarized in Table~\ref{table:Complementation}.
\end{theorem}

\begin{proof}\
\begin{itemize}
\item From weak and B\"uchi. A GFG-NBW $\A$ with $n$ states has an equivalent DBW $\D$ with up to $n^2$ states~\cite{KS15}, on which structure there is a DCW $\overline{\D}$ for the complement language. Notice that $\overline{\D}$ is also a GFG-NCW, GFG-NPW, GFG-NRW, and GFG-NSW. Now, if there is a GFG-NBW equivalent to $\overline{\D}$, then $\overline{\D}$ is DWW-recognizable, and, by Theorem~\ref{thm:GFG-NCW-NBW}, there is a GFG-NWW on a substructure of $\overline{\D}$.

\item From co-B\"uchi. By Corollary~\ref{cor:Complementation}, we have the exponential state blow-up in the complementation to GFG-NPW and GFG-NRW automata. Since the complement of a co-B\"uchi-recognizable language is DBW-recognizable, we get an exponential state blow-up also to GFG-NBW.  

\item To weak and co-B\"uchi. Consider a GFG-NCW, GFG-NPW, or GFG-NRW $\A$ with $n$~states that can be complemented into a GFG-NCW $\C$. Then the language of $\A$ is GFG-NBW recognizable. Thus, by Theorem~\ref{thm:GFG-NRW-NBW}, there is a GFG-NBW equivalent to $\A$ with up to $n$ states. Hence, by case~(1), there is a GFG-NCW for the complement of $\A$ with up to $n^2$ states.

\item  From Streett to weak. Consider a GFG-NSW $\A$ that can be complemented to a GFG-NWW. Then the language of $\A$ is DWW-recognizable. Thus, by Theorems~\ref{thm:NswNcw} and~\ref{thm:GFG-NCW-NBW}, there is a GFG-NWW on a substructure of $\A$, and we are back in case~(1).

\item  From Streett to co-B\"uchi. Given a DRW $\A$ that is NCW realizable, one can translate it to an equivalent NCW by first dualizing $\A$ into a DSW $\overline{\A}$ for the complement language, and then complementing $\overline{\A}$ into a GFG-NCW $\C$. Since dualizing a DRW into a DSW is done with no state blowup and the translation of DRWs to NCWs might involve an exponential state blowup \cite{Bok17}, so does the complementation of GFG-NSW to GFG-NCWs.

\item  From Streett to Streett. Analogous to the above case of Streett to co-B\"uchi, due to the exponential state blowup in the translation of DRWs to NSWs \cite{Bok17}.
\qedhere
\end{itemize}
\end{proof}

\paragraph*{Translating LTL formulas to GFG Automata}

Recall that GFG-NCWs are exponentially more succinct than DCWs~\cite{KS15}, suggesting they do have some power of nondeterministic automata. A natural question is whether one can come up with an exponential translation of LTL formulas to GFG automata, in particular when attention is restricted to LTL formulas that are DCW-realizable. We complete this section with a negative answer, providing another evidence for the deterministic nature of GFG automata. This result is based on the fact that the language with which the doubly-exponential lower bound of the translation of LTL to DBW in~\cite{KV05b} is proven is bounded (that is, it is both safe and co-safe). It means that by Corollary~\ref{cor nww}, any GFG-NSW for it would be \DBP, contradicting the doubly-exponential lower bound.

\newcommand{\corLTLtoGFGNSW}{
The translation of DCW-realizable LTL formulas into GFG-NSW is doubly exponential.
}

\begin{corollary}
\label{cor:ltl-to-gfgnsw}
\corLTLtoGFGNSW
\end{corollary}

\bibliographystyle{plain}

\bibliography{ok}

\end{document}